\documentclass[twoside,11pt]{article}

\usepackage{blindtext}

\usepackage[preprint]{jmlr2e}

\usepackage{placeins}
\usepackage{amsmath}
\usepackage{bm,mathtools,etoolbox,amsfonts,bbm}
\hypersetup{colorlinks,linkcolor={black},citecolor={black},urlcolor={black}}    
\usepackage[nameinlink,capitalise]{cleveref}

\usepackage{adjustbox}
\usepackage{tikz}
\usepackage{circuitikz}
\usetikzlibrary{shapes.geometric, arrows}

\usepackage{setspace}
\usepackage{booktabs}
\usepackage{algorithm}
\usepackage{algpseudocode}
\usepackage{multicol} 
\usepackage{makecell, multirow, tabularx} 
\usepackage{rotating}
\usepackage{float}
\usepackage{ulem}
\usepackage{graphicx}
\usepackage{tabularx}
\usepackage{svg}

\newtheorem{Theorem}{Theorem}[section]
\newtheorem{Proposition}{Proposition}[section]

\newtheorem{Lemma}{Lemma}[section]
\newtheorem{Remark}{Remark}[section]

\newcommand{\ceil}[1]{\left\lceil #1 \right\rceil}
\newcommand{\notcheckmark}{{$\surd$}\textsuperscript{\textcolor{black}{\kern-0.35em{\bf--}}}}
\newcommand{\selfcheckmark}{{$\surd$}\textsuperscript{\textcolor{black}{}}}

\usepackage{lastpage}

\ShortHeadings{Calibration Prediction Interval}{Wu and Politis}
\firstpageno{1}

\begin{document}

\title{Calibration Prediction Interval for Non-parametric Regression and Neural Networks}

\author{\name Kejin Wu \email kwu8@luc.edu\\
       \addr Department of Mathematics and Statistics\\
       Loyola University Chicago\\
       Chicago, IL 60660, USA
       \AND
       \name Dimitris N. Politis \email dpolitis@ucsd.edu\\
       \addr Department of Mathematics and Halicio\u{g}lu Data Science Institute\\
       University of California San Diego\\
       La Jolla, CA 92093, USA}

\editor{My editor}

\maketitle

\begin{abstract}
Accurate conditional prediction in the regression setting plays an important role in many real-world problems. Typically, a point prediction often falls short since no attempt is made to quantify the prediction accuracy. Classically, under the normality and linearity assumptions, the Prediction Interval (PI) for the response variable can be determined routinely based on the $t$ distribution. Unfortunately, these two assumptions are rarely met in practice. To fully avoid these two conditions, we develop a so-called calibration PI (cPI) which leverages estimations by Deep Neural Networks (DNN) or kernel methods. Moreover, the cPI can be easily adjusted to capture the estimation variability within the prediction procedure, which is a crucial error source often ignored in practice. Under regular assumptions, we verify that our cPI has an asymptotically valid coverage rate. We also demonstrate that cPI based on the kernel method ensures a coverage rate with a high probability when the sample size is large. Besides, with several conditions, the cPI based on DNN works even with finite samples. A comprehensive simulation study supports the usefulness of cPI, and the convincing performance of cPI with a short sample is confirmed with two empirical datasets. 
\end{abstract}


\section{Introduction}\label{Sec:Intro}
Accurate prediction conditional on a regressor plays an important role in many real-world applications. Throughout this work, the conditional prediction stands for a complete inference in the regression context, i.e., the point prediction and the corresponding Prediction Interval (PI) for a response variable $Y\in \mathbb{R}$ given some explanatory variable $X_f\in \mathbb{R}^d$. Ideally, if the Cumulative Distribution Function (CDF) $F_{Y|X_f}$ is known, we have many candidates for the point prediction, e.g., mean and median which are optimal w.r.t. $L^2$ and $L^1$ loss criteria, respectively. For PI, we also have various choices, e.g., PI with $1-\alpha$ confidence level symmetric around the mean or one having a minimal length. Unfortunately, $F_{Y|X_f}$ is usually unknown in practice. Therefore, we need to estimate the oracle PI and ensure that the resulting PI is at least asymptotically valid. The asymptotic validity means the coverage rate of PI converges to the nominal level as the sample size approaches infinity. Typically, under the normality and linearity assumptions, the PI for $Y$ possesses an analytic form that involves the $t$ distribution and an appropriate standard error term. However, these two assumptions eliminate the possibility of data being non-linear and non-normal. 

We first propose a non-parametric framework to build PI so that the linearity condition can be eluded. For many real-world problems, it is reasonable to assume that the response variable $Y$ has a normal distribution given corresponding $X_f$ human heights, test scores, etc. In this special case, a $1-\alpha$ oracle symmetric PI can be approximated by
\begin{equation}\label{Eq:PInormal}
    \left[\hat{\mu}_{Y|X_f} - z_{\alpha/2}\sqrt{\hat{\kappa}_{Y|X_f} -  \hat{\mu}^2_{Y|X_f}}~,~ \hat{\mu}_{Y|X_f} + z_{\alpha/2}\sqrt{\hat{\kappa}_{Y|X_f} -  \hat{\mu}^2_{Y|X_f}}  \right];
\end{equation}
$z_{\alpha/2}$ is the right-tail $\alpha/2$ quantile value of a standard normal random variable; $\hat{\kappa}_{Y|X_f}$ is an estimator of $\mathbb{E}(Y^2|X_f)$; $\hat{\mu}_{Y|X_f}$ is an estimator of $\mathbb{E}(Y|X_f)$. We should notice that there is no parametric form being assumed about the underlying data-generating model, so PI (\ref{Eq:PInormal}) is compatible with linear/non-linear data. Moreover, it is asymptotically valid if $\hat{\kappa}_{Y|X_f}$ and $\hat{\mu}_{Y|X_f}$ are consistent estimators under the normality assumption. Throughout this paper, we take PI (\ref{Eq:PInormal}) to be the benchmark method.

Another non-parametric way to construct PI relies on $\widehat{F}_{Y|X_f}(\cdot)$ which is an estimation of the function $F_{Y|X_f}(\cdot)$. Then, a PI can be constructed as follows:
\begin{equation}\label{Eq:QPI}
    \left[ \widehat{F}^{-1}_{Y|X_f}(\alpha/2) ; \widehat{F}^{-1}_{Y|X_f}(1-\alpha/2)\right];
\end{equation}
The PI (\ref{Eq:QPI}) can be derived from the oracle equal-tailed PI with the true CDF $F_{Y|X_f}(\cdot)$ being replaced by the estimated one $\widehat{F}_{Y|X_f}(\cdot)$. In case that $\widehat{F}_{Y|X_f}(\cdot)$ is not continuous and strictly increasing w.r.t. $y$ so that the quantiles are not well-defined, we define $\widehat{F}^{-1}_{Y|X_f}(\alpha) = \inf\{y: \widehat{F}_{Y|X_f}(y) \geq \alpha \}$. Under the assumption that $\widehat{F}_{Y|X_f}(\cdot)$ converges to $F_{Y|X_f}(\cdot)$ consistently at $\alpha$ and $1-\alpha$ for all $X_f = x_f$ and $y$ , it is easy to verify PI (\ref{Eq:QPI}) is asymptotically valid with a $1-\alpha$ confidence level when $F_{Y|X_f}$ is continuous at $F^{-1}_{Y|X_f}(\alpha/2)$ and $F^{-1}_{Y|X_f}(1-\alpha/2)$. A classical candidate for $\widehat{F}_{Y|X_f}(\cdot)$ is based on kernel estimation. Then, linearity and normality assumptions are not necessary anymore. In addition, the quantile regression approach can be an alternative that suffers less from the curse of dimensionality compared to kernel methods when $d$ is large. 

Unfortunately, PI (\ref{Eq:QPI}) usually undercovers $Y$ given $X_f$ inevitably with finite samples since the estimation variability inherent from $\widehat{F}_{Y|X_f}(\cdot)$ is not captured. By our simulation studies, it turns out that the PI (\ref{Eq:PInormal}) undercovers $Y$ even with the normality assumption and a large sample size. This finding highlights the necessity of capturing estimation variability in the PI; otherwise, undercoverage issues will prevail. To correct the undercoverage of PI due to uncaptured estimation uncertainties, \cite{politis2015model} defined a so-called \textit{pertinent} PI (PPI) which can capture the estimation variability in finite sample cases, but did not provide the coverage guarantee. Later, \cite{wang2021model} developed the PPI with standard CDF estimators in the context of regression. It was revealed that the PPI is more accurate and stable in contrast to PI (\ref{Eq:QPI}) and conformal prediction of \cite{lei2018distribution}; \cite{wu2024deep} further constructed PPI for regression with Deep Neural Networks (DNN) and showed it is superior than other two deep generative methods \cite{zhou2023deep} and \cite{liu2021wasserstein}. 

According to folk wisdom and empirical/theoretical evidence, DNN estimators suffer less from the curse of dimensionality than kernel estimators. Thus, it is appealing to utilize the DNN in the prediction procedure.  Nevertheless, the above-mentioned DNN-based PI depends on a tuning parameter severely, which renders unavoidable complexity in practice. In this paper, we propose a calibration approach to make prediction intervals with either DNN or kernel-based estimators. We should mention that the word ``calibration'' in our paper is different from the process of adjusting the predicted probabilities of a model so that they better reflect the true likelihood of an event; see \cite{guo2017calibration} for a related discussion. More precisely, our calibration method stands for a specifically designed manipulation to extend the PI to capture the estimation variability so that the undercoverage issue can be compensated for.

For the DNN-based approach, we consider the standard fully connected forward DNN with ReLU activation functions (abbreviated to DNN hereafter) whose structure is determined by the dimension of input, the depth of the hidden layer and the width for each layer. Our calibration method can be incorporated with DNN without a crucial parameter-tuning step. Although the number of grid points is a required hyperparameter in the setup procedure of calibration PI (cPI), the choice of this number does not affect the performance of PI too much compared to the tuning parameter required in deep generative methods; this fact is revealed by real-data studies later. In addition, this calibration approach is easily adaptive to parallel computation so that the computational time is not necessarily longer than building a single PI (\ref{Eq:PInormal}) or (\ref{Eq:QPI}). In short, the theoretical foundation of the DNN-based cPI hinges on the universal approximation ability of DNN. 

Routinely, people treat DNN as a ``black box'' that takes some inputs and returns desired outputs. The uncertainty quantification of outputs is generally ignored. Instead of leaning on the generative idea applied in the existing work, our cPI provides a simple way to quantify the prediction uncertainties and this procedure can be extended to scenarios with other types of neural networks. Moreover, the cPI can compensate for the estimation variability, which is usually not captured in practice. We show that the cPI with DNN estimators is asymptotically valid. In addition, it can guarantee the coverage rate even for a finite-sample case under some conditions. 

We can also combine the calibration approach with kernel-based estimators. The main idea is that we first estimate $f_{Y|X}$ which is the conditional density of $Y$ given $X$ by $\widehat{f}_{Y|X}$ and then the cPI is built based on some specific calibration approach and the CDF estimation which is obtained by taking the integration on $\widehat{f}_{Y|X}$. Theoretically, the asymptotic validity of cPI is trivial since $\widehat{f}_{Y|X}$ is typically consistent with the true conditional density function pointwise. Furthermore, we show that the cPI with kernel estimator can guarantee the coverage rate when the sample size is larger than a constant $N$ with a high probability. This non-asymptotic result is more delicate than the asymptotic validity property.

The paper is organized as follows. We introduce the methodology and algorithm for our cPI with DNN in \cref{Sec:CaliPIMethod}. The variant with kernel estimators is discussed in \cref{Sec:CaliPIMethodKernel}. Different calibration approaches are discussed in \cref{Sec:DiffCali}. The theoretical validation of our methods is presented in \cref{Sec:Theorey}. Simulation and empirical data studies are conducted in \cref{Sec:Simulation} and \cref{Sec:RealData}, respectively. \cref{Sec:Conclusion} concludes this paper.

\subsection{Our Contributions}
Our main contributions are as follows:
\begin{itemize}
    \item[1.] We propose a calibration idea to capture the estimation variability, which is usually ignored by users when a Prediction Interval (PI) is built. It is well known that the PI has a prevailing undercoverage issue in practice. Our calibration method can compensate for the empirical undercoverage and it can be incorporated with PI based on various estimators, such as DNN and kernel estimators. As different from other approaches to make PI in a deep generative way, our method can be performed simply without sophisticated tuning efforts. Also, our method can be deployed in parallel, so the calibration process does not bring a significant increase in the computation time.
    \item[2.] We develop the theory to show that our calibration Prediction Intervals (cPI) with DNN and kernel estimators are both asymptotically valid. Moreover, we show that the cPI based on kernel estimators can guarantee the coverage rate with a high probability when the sample size is larger. This property is more delicate than the asymptotic validity property. Due to the universal approximation property of DNN, under several conditions, we further show that the cPI based on DNN can return a coverage rate larger than or equal to the nominal confidence level even for a finite sample size. As far as we know, this is the first attempt to discuss the finite sample coverage property in this context. 
    \item[3.] We propose a new calibration method to correct the monotonicity of CDF estimators. From simulation and real data analyses, this new correction methodology works better than other variants according to the subsequent performance of the resulting PI. This correction method stands for its own interests in the estimation of CDF and its correction.
    \item[4.] We verify the importance of capturing the estimation variability in the procedure to build PI by simulation studies. Otherwise, the standard PI (\ref{Eq:PInormal}) is spoiled by undercoverage even when the data is normal. Also, we find the cPI with DNN tends to work better than the cPI with kernel estimators in simulations and real data analysis; this can be seen as new evidence to support the folk wisdom that DNN estimators suffer less from the curse of dimensionality than kernel estimators.

\end{itemize}

\subsection{Related Works}
From the perspective of building a prediction interval with neural networks estimators, there are various approaches. For example,  (1) \cite{liu2021wasserstein, zhou2023deep} considered deep generative approaches to estimate the whole distribution of future response variable conditional on the new predictor variables. Their methods can return asymptotically valid prediction intervals since all methods can estimate the conditional distribution of future response variable consistently; (2) If the training sample size is ample, \cite{wu2024scalable} proposed a subagging estimator of DNN and then built the PI under the assumption of additive structure of the regression model; (3) Within a residual-based framework, \cite{padilla2024confidence} considered the conditional variance estimation of future response variable by DNN. Then, the PI can be built in a similar format of PI (\ref{Eq:PInormal}); (4) Instead of relying on subagging or bootstrapping, \cite{mancini2020prediction} considered determining the PI in an ensemble way motivated by a so-called Extra-trees algorithm; \cite{simhayev2020piven} also took the ensemble idea with their specific DNN structure; (5) In a model-free prediction principle, \cite{politis2015model} proposed a so-called pertinent PI (PPI) and then \cite{wu2024deep} achieved this PPI with DNN; (6) There are many works in which the PI based on DNN are taken as some kind of output of DNN and then returned directly after training DNN with specific loss function; see \cite{khosravi2010lower,pearce2018high,tagasovska2019single} for references. Neither of these approaches discusses the coverage property of PI for the finite sample case. Some of them also require a complicated procedure or rely crucially on tuning parameters. On the other hand, our calibration method is simple to deploy with any estimators. Once these estimators are consistent, our calibration PI is asymptotically valid. Under some conditions, our calibration PI can even guarantee the coverage for finite sample cases. 

Building a PI based on non-parametric estimators is a standard task in Statistics. Participants can always determine the lower and upper bounds of a PI by approximating the conditional quantile of the future response variable; see \cite{chaudhuri1991nonparametric,takeuchi2006nonparametric} for some references; see \cite{koenker2017quantile} for a comprehensive review of quantile regression. In this work at hand, our calibration method can be understood as a reverse approach of building a quantile-regression-based PI. In short, we fix some values with unknown quantile levels and then calibrate the PI to cover $Y$ given $X_f$. Due to such a reverse approach, our cPI can relieve the adverse effects of estimation variability on the coverage rate of the prediction interval.

For other general methods to build PI, such as Bayesian technique and Conformal prediction strategy, we refer readers to the review paper of \cite{tian2022methods}. In this work at hand, we intend to provide a simple and unified procedure to capture the estimation variability in PI, especially for DNN and kernel estimators. Compared to most relevant PIs, our cPI has an overall better performance regarding to the coverage rate and length.

\section{Calibration PI with DNN}\label{Sec:CaliPIMethod}
Throughout this paper, we focus on the prediction inference of the regression model $Y = G(X,\epsilon)$; $\epsilon$ is the error term with distribution $P_{\epsilon}$ and density $f_{\epsilon}$. We assume the underlying probability space is $(\Omega,\mathcal{F},\mathbb{P})$. Furthermore, we assume:
\begin{itemize}
    \item[A1] The regression model $G(X,\epsilon)$ is a uniformly continuous function w.r.t. $X\in \mathbb{R}^d$ and $\epsilon$ and $Y\in \mathbb{R}$ belongs to $L^2(\mathbb{P})$; 
    \item[A2] $X$ and $\epsilon$ are independent. The error density $f_{\epsilon}$ is bounded;
    \item[A3] The conditional cumulative distribution $F_{Y|X}(y)$ is a continuous function w.r.t. $X$ for all points $y$ in $\mathcal{Y}$ which is the domain of $Y$.
\end{itemize}
\begin{Remark}
    In the context of estimation by DNN, the smoothness of the true functions affects the error rate of estimation. Thus, the above assumptions give the weakest condition regarding the smoothness of $G(\cdot,\cdot)$, $F_{Y|X}(y)$, etc. In practice, these functions may be smoother than what we assume here.
\end{Remark}

\subsection{Preliminary Analysis}
Before explaining the logic of our cPI with DNN, we consider the idea to determine the optimal $L^2$ point prediction of $Y$ given ``future'' $X_f$ as an example. Under A1, it is well known that $\mathbb{E}(Y-\mathbb{E}[Y|X])^2 \leq \mathbb{E}(Y-h(X))^2$ for any function $h:\mathcal{X}\to\mathbb{R}$ and the equality holds if and only if $h(X)=\mathbb{E}(Y|X)$. We call $\mathbb{E}(Y|X)$ the conditional expectation of $Y$ given $X$. More specifically, it is conditional on $\sigma(X)$, namely the $\sigma$-algebra generated by $X$. Thus, $\mathbb{E}(Y|X_f)$ is the desired optimal $L^2$ prediction given $X_f$. Revealed by Radon–Nikodym theorem, the conditional expectation always exists and is unique almost everywhere, i.e., $\mathbb{E}(Y|X)$ is actually an equivalence class of random variables that are $\sigma(X)$-measurable and satisfy $\int_A \mathbb{E}(Y|X) d \mathbb{P}=\int_A Y d \mathbb{P}$, for all $A \in \sigma(X)$. For simplification, we think of $\mathbb{E}(Y|X)$ as a single function rather than a class of functions. To find the conditional expectation of $G(X,\epsilon)$ in our setting, we can solve the following optimization problem:

\begin{equation}\label{Eq:MSEtrue}
    H_0 = \arg\min_{H}\mathbb{E}(G(X,\epsilon) - H(X))^2;
\end{equation}
where the optimizer $H_0$ is searched from all measurable functions $H: = \mathcal{X} \to \mathbb{R}$. We can prove that $H_0(X) = \mathbb{E}(G(X,\epsilon)|X)$. Going one step further, we can conclude that $H_0(x) = \mathbb{E}(G(x,\epsilon))$ in which the expectation on the r.h.s. is only w.r.t. $\epsilon$. Moreover, $ H_0(x) = \arg\min_{H}\mathbb{E}(G(x,\epsilon) - H(x))^2$. In addition, the continuity of $H_0(x)$ can be verified under A2. We summarize these discussions into \cref{Lemma:CDE}.
\begin{Lemma}\label{Lemma:CDE}
    Under A1 and A2, the optimizer in \cref{Eq:MSEtrue} is $H_0(X) = \mathbb{E}(G(X,\epsilon))$. Moreover, $H_0(x) = \mathbb{E}(G(x,\epsilon))$ is continuous w.r.t. the argument $x$. 
\end{Lemma}
In the context of estimation by DNN, we define $H_{\text{DNN}}: \mathcal{X}\to\mathbb{R}$ by the equation below: 
\begin{equation}\label{Eq:MSEDNN}
    H_{\text{DNN}} = \arg\min_{H_\theta\in\mathcal{F}_{\text{DNN}}}\mathbb{E}(G(X,\epsilon) - H_\theta(X))^2;
\end{equation}
$\theta$ represents all parameters involved in a DNN; $\mathcal{F}_{\text{DNN}}$ is a user-chosen class of DNNs whose complexity increases as the sample size becomes larger and larger. In practice, the population-level objective $\mathbb{E}(G(X,\epsilon) - H(X))^2$ is unattainable. Thus, we strive to find a consistent estimator to approximate $\mathbb{E}(Y|X_f)$ based on the empirical version of $H_{\text{DNN}}$, namely $\widehat{H}$ defined below:
$$
\widehat{H} = \arg\min_{H_\theta\in\mathcal{F}_{\text{DNN}}}\frac{1}{n}\sum_{i=1}^n(Y_i - H_\theta(X_i))^2;
$$
assuming we have data pairs $\{(Y_i,X_i)\}_{i=1}^n$ at hand. The convergence behavior of $\widehat{H}$ to $H_0$ has been widely studied in different scenarios; see \cite{jiao2023deep, farrell2021deep, nakada2020adaptive} for some references. In short, a non-asymptotic error bound can be built to guarantee $\|H_0 - \widehat{H}\|_{L^2}$ converges to 0 as the sample size $n\to \infty$ with a high probability; see \cref{Sec:Theorey} for more discussions. In other words, $\widehat{H}(X_f)$ can approximate the optimal $L^2$ point prediction of $Y$ given $X_f$ empirically.

\subsection{The Main Idea of Calibration PI}
A point prediction often falls short of being satisfactory. A prediction interval (PI) is necessary to quantify the prediction uncertainty similarly to the usage of CI in the estimation procedure. As exuded in \cref{Sec:Intro}, naive PIs (\ref{Eq:PInormal}) and (\ref{Eq:QPI}) are not desirable due to the strict normality assumption or the disability to capture the estimation uncertainties. Unlike the existing DNN-based generative methods in which some tuning parameters are crucial, we consider a self-contained calibration method to build the PI.

Our cPI is motivated by the universal estimation property of DNN and kernel methods on approximating smooth functions. In addition, \cref{Lemma:CDE} manifests that the conditional expectation can be a continuous function of the given independent variable $X$. Therefore, it is straightforward to use a DNN or a kernel estimator to approximate conditional expectations. In this spirit, we first define a set of random variables $Z_j := \mathbbm{1}(Y\leq q_j)$ for $j = 1,\ldots, g$; $\mathbbm{1}(\cdot)$ is the indicator function; $q_j$ is the $j$-th fixed point in the range of $Y$ assuming this range is a compact set. The collection of all $q_j$ is a set of grid points spanned in $\mathcal{Y}$ with equal distance $D$. The effects of the distance $D$ on the coverage rate will be discussed in \cref{Sec:Theorey} theoretically. In practice, $Z_{j,i}: = \mathbbm{1}(Y_i\leq q_j)$ are available for $j = 1,\ldots,g$ and $i = 1,\ldots, n$ once we observe samples $\{(Y_i,X_i)\}_{i=1}^n$. Correspondingly, we initialize $g$ number of DNNs, namely $\{H_{\theta,j}\}_{j=1}^g$, that belong to a pre-specified $\mathcal{F}_{\text{DNN}}$. For each $j$, we train $H_{\theta,j}$ based on below objective function:
\begin{equation}
    \widehat{H}_j = \arg\min_{H_{\theta,j}\in\mathcal{F}_{\text{DNN}}}\frac{1}{n}\sum_{i=1}^n(Z_{j,i} - H_{\theta,j}(X_i))^2;
\end{equation} 
obviously, $\widehat{H}_j$ is the empirical version of $H_{\text{DNN,j}}$ which minimizes $\mathbb{E}(Z_j - H_{\theta,j}(X))^2$. Additionally, the following lemma depicts the property of $H_j : =  \arg\min_{H}\mathbb{E}(Z_{j} - H(X))^2$:
\begin{Lemma}\label{Lemma:EstHk}
Under A1 to A3, $H_j(X) = \mathbb{E}(\mathbbm{1}(Y\leq q_j)|X) = \mathbb{E}(\mathbbm{1}(G(X,\epsilon)\leq q_j))$. Besides, $H_j(x)$ is a continuous function w.r.t. $x$. 
\end{Lemma}
 
Intuitively, $\widehat{H}_j$ is pretty ``close'' to $H_j$ in the mean square error sense when the sample size is large. The convergence result will be developed in \cref{Sec:Theorey}. Particularly, $\widehat{H}_j(X_f)$ is an estimation of $F_{Y|X_f}(q_j)$. Collecting all trained $\{\widehat{H}_j\}_{j=1}^g$, we get a comprehensive evaluation of the CDF $F_{Y|X_f}(\cdot)$ at grid points $\{q_j\}_{j=1}^g$. Depending on the accessible computational source, $g$ can be a large number to refine the approximation of $F_{Y|X_f}(\cdot)$ by $\{\widehat{H}_j(X_f)\}_{j=1}^g$. However, these approximations are immature unless some specific correction procedure is manipulated to guarantee that $\widehat{H}_1(X_f)\leq \widehat{H}_2(X_f)\leq \ldots \leq \widehat{H}_g(X_f)$ under the condition $q_1\leq q_2\leq \ldots \leq q_g$ for some specific $X_f$; see \cref{Subsec:CorrectionforMono} about correction for monotonicity. Consequently, a PI can be built in a calibration way relying on the corrected $\{\widehat{H}_j(X_f)\}_{j=1}^g$. For example, a PI with $1-\alpha$ confidence level can be made by searching two indices $l$ and $r$ among $1,\ldots,g$ to make sure $\widehat{H}_r(X_f) - \widehat{H}_l(X_f)\geq 1-\alpha$. Correspondingly, two endpoints of our cPI can be determined by $q_l$ and $q_r$; see \cref{Sec:DiffCali} for more details about different calibration approaches. In theory, the cPI is asymptotically valid as $n\to\infty$ and $g\to\infty$. In practice with finite or large samples, the cPI offers a systematic way to alleviate the undercoverage issue, i.e., the calibration step is an alternative way to compensate for the usually omitted estimation variability. The good performance of cPI is verified by simulated and real-world data. 

\begin{Remark}
    In practice, we may want to build a PI conditional on a specific $X_f= x_f$, i.e., we treat the future regressor as a fixed scalar or vector. From \cref{Lemma:CDE} and \cref{Lemma:EstHk}, we know that $H_0(x)$ and $H_j(x)$ for $j =1, \ldots, g$ are all continuous w.r.t. to $x$. Later, under additional mild conditions, we will show that consistent estimators of $H_0(x_f)$ and $H_j(x_f)$ for all possible values of $x_f$ can be acquired no matter with DNN or kernel estimation techniques. As a result, our cPI is asymptotically valid conditional on $X_f= x_f$. With a slight abuse of notation, $X_f$ hereafter also stands for a specific $x_f$. 
\end{Remark}

\begin{Remark}[Comparison of our cPI and PI by quantile regression]
At first glance, it seems that our cPI is similar to PI (\ref{Eq:QPI}) based on quantile estimation. However, these two approaches are essentially different. For cPI, instead of estimating the quantile of $Y$ given a new $X_f$, we take a reverse approach and try to estimate the $F_{Y|X_f}(q_j)$ for all $\{q_j\}_{j=1}^g$. In particular, we fix some values with unknown quantile levels and then calibrate the PI to cover $Y$ given $X_f$. By taking such a reverse approach, the first advantage is that the risk is based on mean squared error rather than quantile loss which is more appropriate for finding estimators based on the gradient descent technique; the second advantage is that our cPI can relieve the adverse effects of estimation variability on the coverage rate of the prediction interval.  
\end{Remark}

All in all, the construction of our cPI consists of three main steps: (1) estimation of $F_{Y|X_f}$ at all grid points by DNNs or kernel estimators; (2) correction of estimators for monotonicity; (3) determination of PI with an appropriate calibration to guarantee a lower bound of coverage rate. 

\subsection{Algorithm of Calibration PI}
We summarize the procedure to build cPI with DNN in \cref{Algo:CaliPIwithDNN}; where the learning rate $lr$ is a hyperparameter that controls the step size to update the parameters of a DNN; the batch size $B$ is the number of training samples that are processed together during a single iteration of training; the number of epochs $E$ is the number of times that the entire training dataset is passed through the model during the training process; the clipping parameter $m$ is a positive constant which restricts the value region of the weights and bias in a DNN to be $[-m,m]$, which means any parameters with updated values larger than $m$ or smaller than $-m$ will be truncated to $m$ and $-m$, respectively; the calibration approach used to build cPI will be introduced in \cref{Sec:DiffCali}; the CDF correction method will be discussed in \cref{Subsec:CorrectionforMono}. The choice of these hyperparameters will be specified in \cref{Sec:Simulation,Sec:RealData}. The steps to constructing cPI with kernel methods are described in \cref{Algo:CaliPIwithKernel}.

\begin{Remark}[Comments of \cref{Algo:CaliPIwithDNN}]
The training of $\widehat{H}_{0}, \widehat{H}_{1},\ldots,\widehat{H}_{g}$ can be performed parallelly to save running time. To build the naive PI (\ref{Eq:PInormal}), we need to get $\widehat{H}_{0}$ and one more DNN to estimate $\mathbb{E}(Y^2|X_f)$. Thus, with parallel training, the computation time of building cPI could be the same as fitting a single DNN. 
\end{Remark}

\begin{algorithm}[htbp]
\caption{The outline to build calibration PI with DNN}\label{Algo:CaliPIwithDNN}
\setstretch{1.25}
\begin{algorithmic}[1]
\Require (a) the training samples $\{X_i,Y_i\}_{i=1}^n$; (b) the new observed $X_f$; (c) the learning rate $lr$; (d) the batch size $B$; (e) the number of epochs $E$; (f) the clipping parameter $m$; (g) the number of grid points $g$; (h) the specific optimizer, e.g., SGD, Adam or RMSProp; (i) the calibration approach; (j) the CDF correction method; (k) the significance level $\alpha$.
\Statex

\State Initialize $g+1$ number of DNN $H_{\theta,j}(X)\in \mathcal{F}_{\text{DNN}}$ for $j = 0,\cdots, g$; $\mathcal{F}_{\text{DNN}}$ is a DNN class with some structure constraint. Without loss of generality, we train $H_{\theta,0}$ to estimate $\mathbb{E}(Y|X)$.
\For{number of epochs} 
\State Update $H_{\theta,0}$ by descending its stochastic gradient with the chosen optimizer on batch data:
$$
\nabla_{\theta} \left\{ \frac{1}{n}\sum_{i=1}^n\left( Y_i - H_{\theta,0}(X_i)  \right)^2   \right\}.
$$
\State Clip the parameter of $H_{\theta,0}$ to $[-m,m]$.
\EndFor
\State \textbf{Get} trained $\widehat{H}_{0}$.
\State Take $g$ number of points to split the region $[\min(Y_1,
\ldots, Y_n), \max(Y_1,
\ldots, Y_n)]$ with equal gap. Denote these points by $q_1,\ldots,q_g$. 
\State Define $g$ number of new random variables $Z_j := \mathbbm{1}(Y\leq q_j)$. 
\For{j in (1,\ldots, $g$)}
\State Based on training samples $\{X_i,Y_i\}_{i=1}^n$, compute observed value $Z_{j,1},\ldots, Z_{j,n}$. 
\For{number of epochs} 
\State Update $H_{\theta,j}$ by descending its stochastic gradient with the chosen optimizer on batch data:
$$
\nabla_{\theta} \left\{ \frac{1}{n}\sum_{i=1}^n\left( Z_{j,i} - H_{\theta,j}(X_i)  \right)^2 \right\}. 
$$
\State Clip the parameter of $H_{\theta,j}$ to $[-m,m]$.
\EndFor
\EndFor
\State \textbf{Get} trained $\widehat{H}_{1},\ldots,\widehat{H}_{g}$.
\State Compute $\widehat{H}_{0}(X_f), \widehat{H}_{1}(X_f),\ldots,\widehat{H}_{g}(X_f)$ for a specific future $X_f$ of interests; $\widehat{H}_{0}(X_f)$ is necessary to build a symmetric PI. 
\State If needed, apply one method to correct values of $\{\widehat{H}_{1}(X_f),\ldots,\widehat{H}_{g}(X_f)\}$ for monotonicity; then interpolate the corrected $\{\widehat{H}_{1}(X_f),\ldots,\widehat{H}_{g}(X_f)\}$ by linear segments to estimate the conditional CDF of $Y$ given $X_f$.  
\State Take one calibration approach to find a PI, namely $\widehat{I}_{C}$, that is designed to cover $Y$ given $X_f$ with at least $(1-\alpha)$ confidence level. 
\State \textbf{Return} $\widehat{I}_{C}$. 
\end{algorithmic}
\end{algorithm}


\section{Calibration PI with Kernel Estimators}\label{Sec:CaliPIMethodKernel}
Besides the cPI with DNN estimators, we can consider a cPI relying on the estimated value of CDF evaluated at $q_j$ given $X_f$ based on a standard kernel estimator. Let's assume that we have built the appropriate kernel estimator $\widehat{F}_{Y|X_f}$, and we can evaluate it at the same grid points we used to determine the cPI with DNN. Then, we can build a cPI according to the same procedure that was applied before. We should remark that the kernel estimator has been used to compare the performance of conditional point prediction with deep generative methods in \cite{wu2024deep, zhou2023deep, liu2021wasserstein}. However, they did not consider building PI based on kernel estimators. We fulfill this empirical gap with simulations and real-data studies. Also, we develop the theory to show the asymptotic validity of the corresponding cPI. Moreover, we show that the lower-bound coverage rate can be guaranteed with a high probability; see the theoretical analysis presented in \cref{Sec:Theorey}. Notably, the real-data studies with short training size disclose that the cPI with DNN generally outperforms cPI based on kernel estimators; this may be another empirical evidence to support the common thought that the DNN estimator suffers less from the curse of dimensionality than kernel estimators; see more details from \cref{Sec:RealData}.

\subsection{The Algorithm to Build Calibration PI with Kernel Estimators}
In this paper, we consider the standard conditional density estimator $\hat{f}_{Y|X}(y|x): = \hat{f}_{X,Y}(x,y)/
\hat{f}_X(x)$, where $\hat{f}_{X,Y}(x,y)$ and $\hat{f}(x)$ are joint density estimator of $(X,Y)$ and the marginal density estimator of $X$, respectively, i.e.,
\begin{equation}\label{Eq:Kernelestimators}
    \hat{f}_{X,Y}(x,y)= \frac{1}{n|\Lambda_x|h_0} \sum_{i=1}^n K\left(\Lambda_x^{-1}(x-X_i)\right) k_0\left(\frac{y-Y_i}{h_0}\right)~;~ \hat{f}_X(x)=\frac{1}{n|\Lambda_x|} \sum_{i=1}^n K\left(\Lambda_x^{-1}(x-X_i)\right);
\end{equation}
where $K(\cdot)$ is a multivariate kernel function; $k_0(\cdot)$ is one univariate kernel functions; $\Lambda_x$ represents the bandwidth matrix; $|\Lambda_x|$ is the determinant of $\Lambda_x$; $x$ and $X_i$ are two $d$-dimension vectors. To simplify the analysis, we assume all independent variables are continuous and apply the product kernel with $\Lambda_x$ being a diagonal matrix with a common bandwidth, i.e., assuming $\Lambda_x = h\cdot I$ and $K\left(\Lambda_x^{-1}(x-X_i)\right)  =\prod_{s=1}^d k\left(\frac{x_s-X_{i,s}}{h}\right)$; $I$ is the identity matrix; $h$ is the common bandwidth;
$X_{i,s}$ and $x_{s}$ denote the $s$ th component of $X_i$ and $x$ for $s=1, \ldots, d$; $k(\cdot)$ is a univariate kernel which satisfies below four conditions:
\begin{equation}\label{Eq:fourconditionsonKernel}
k(v)\geq 0, \forall v;~\int k(v) d v=1;~k(v)=k(-v);~\int v^2 k(v) d v=\kappa_2>0. 
\end{equation}
To ensure the asymptotically valid cPI based on kernel estimators, any univariate kernel function satisfying the above conditions can be used. To simplify the analysis, we also assume that the dependent variable is continuous. Thus, we can apply the same univariate kernel function for $k(\cdot)$ and $k_0(\cdot)$. In other words, we can express $\hat{f}(x,y)$ as:
$$
\hat{f}_{X,Y}(x,y) = \frac{1}{nh^{d+1}} \sum_{i=1}^n \widetilde{K}\left(\Lambda^{-1}(z-Z_i)\right);
$$
where $z$ and $Z_i$ represent the vector $(x^T,y)^T$ and random vector $(X_i^T, Y_i)^T$, respectively; $\Lambda$ represents the diagonal matrix with a common bandwidth; $\widetilde{K}(\Lambda^{-1}(z-Z_i)) : = \prod_{s=1}^{d+1}  k\left(\frac{z_s-Z_{i,s}}{h}\right)$; $Z_{i,s}$ and $z_{s}$ denote the $s$-th component of $Z_i$ and $z$ for $s=1, \ldots, d+1$. If there are some discrete independent or dependent variables involved in the regression, we can apply specifically designed kernels; see Section 4 of \cite{li2007nonparametric} for more details. To select the bandwidth $h$, we can consider the cross-validation approach. The algorithm of building cPI with kernel estimators is similar to \cref{Algo:CaliPIwithDNN} and is summarized in \cref{Algo:CaliPIwithKernel}. After determining one calibration method from the different choices presented in \cref{Sec:DiffCali}, we can get various cPI with kernel estimators. To simplify the notations, we take $\hat{f}(y|x)$, $\hat{f}(x,y)$ and $\hat{f}(x)$ to represent $\hat{f}_{Y|X}(y|x)$, $\hat{f}_{X,Y}(x,y)$ and $\hat{f}_X(x)$, respectively. 

\begin{Remark}[Comments of \cref{Algo:CaliPIwithKernel}]
In Step 3, the correction for the monotonicity is not necessary since $\hat{f}(y|x)$ is always positive due to conditions (\ref{Eq:fourconditionsonKernel}) imposed on univariate kernel functions $k(\cdot)$ and $k_0(\cdot)$.
\end{Remark}

\begin{algorithm}[h]
\caption{The outline to build calibration PI with kernel estimators}\label{Algo:CaliPIwithKernel}
\setstretch{1.25}
\begin{algorithmic}[1]
\Require (a) the training samples $\{X_i,Y_i\}_{i=1}^n$; (b) the new observed $X_f$; (c) the number of grid points $g$; (d) the specific bandwidth-choice strategy, e.g., maximum likelihood or least squares cross-validation; (e) the calibration approach; (f) the significance level $\alpha$.
\Statex

\State Apply the bandwidth-choice strategy to determine the kernel estimators $\hat{f}(x,y)$ and $\hat{f}(x)$. Then, the conditional density estimator $\hat{f}(y|x)$ is obtained. 
\State For a future $X_f$ of interest, compute $\widehat{H}^k_{0}(X_f)$ as
$$
\int y\hat{f}(y|X_f) dy;
$$
compute $\widehat{H}^k_{1}(X_f),\ldots,\widehat{H}^k_{g}(X_f)$ as
$$
\int_{-\infty}^{q_j}\hat{f}(y|X_f) dy~\text{for}~j = 1,\ldots,g,~\text{respectively};
$$
$\widehat{H}^k_{0}(X_f)$ is necessary to estimate the conditional mean so that a PI which is symmetric around the mean can be acquired.
\State Interpolate $\{\widehat{H}^k_{1}(X_f),\ldots,\widehat{H}^k_{g}(X_f)\}$ by linear segments to estimate the conditional CDF of $Y$ given $X_f$. 
\State Take the calibration approach to find a PI, namely $\widehat{I}^k_{C}$, that is designed to cover $Y$ given $X$ with at least $(1-\alpha)$ confidence level. 
\State \textbf{Return} $\widehat{I}^k_{C}$. 
\end{algorithmic}
\end{algorithm}


\section{Calibration Approaches}\label{Sec:DiffCali}
Once we have corrected the CDF estimation with DNN for monotonicity or obtained the appropriate CDF estimation from computing the integral based on the kernel density estimator, we can consider different calibration approaches to construct PI with a designed (at least) $1-\alpha$ confidence level. Depending on the correction methods we applied in the previous step and the structure of PI we hope to adopt, there are several PI candidates we can rely on. We mainly present the algorithms of different types of cPIs with DNN. The corresponding cPIs with kernel estimators can be built similarly; see \cref{Algo:CaliPIaak} for a specific illustration.

From the discussion of correction methods in \cref{Subsec:CorrectionforMono}, we have three methods (C1 to C3) to correct for the monotonicity of our CDF estimations of $\{\widehat{H}_j(X_f)\}_{j=1}^g$. We denote three resulting corrected CDF estimators $\widetilde{F}_{\text{Avg}}(\cdot)$, $\widetilde{F}_{\text{RtoL}}(\cdot)$ and $\widetilde{F}_{\text{LtoR}}(\cdot)$; see their definitions in \cref{Subsec:CorrectionforMono}. 

Subsequently, we have several algorithms to build the cPIs. The first one is the cPI with a minimal length based on $\widetilde{F}_{\text{Avg}}$, namely $PI_m$. The main idea is to search two indices $l$ and $r$ from $1,\ldots,g$ such that $\widetilde{F}_{\text{Avg}}(q_r) - \widetilde{F}_{\text{Avg}}(q_l) \geq 1 - \alpha$ and minimize $q_r - q_l$ meanwhile. We summarize the corresponding algorithm in \cref{Algo:CaliPIm}.

\begin{algorithm}
\caption{Calibration PI with DNN in a minimal length: $PI_m$}\label{Algo:CaliPIm}
\setstretch{1.25}
\begin{algorithmic}[1]
\Require All required arguments in \cref{Algo:CaliPIwithDNN}.
\State Initialize $g$ number of DNN $H_{\theta,j}(X)\in \mathcal{F}_{\text{DNN}}$ for $j = 1,\cdots, g$; $\mathcal{F}_{\text{DNN}}$ is a DNN class with some structure constraint. 

\State Perform Steps 7 - 16 in \cref{Algo:CaliPIwithDNN} to get trained $\widehat{H}_{1},\ldots,\widehat{H}_{g}$. Then, compute $\widehat{H}_{1}(X_f),\ldots,\widehat{H}_{g}(X_f)$ for a future $X_f$ of interest and take the correction approach C3 in \cref{Subsec:CorrectionforMono} to get $\widetilde{F}_{\text{Avg}}$.
\State Solve the below optimization problem to determine two indices $l$ and $r$:
    \begin{equation*}
        \begin{split}
        &l, r = \arg\min_{l,r\in\{1,\ldots,g\}}(q_r - q_l);\\
        &\text{such that},~\widetilde{F}_{\text{Avg}}(q_r) - \widetilde{F}_{\text{Avg}}(q_l) \geq 1-\alpha.
        \end{split}
    \end{equation*}

\State \textbf{Return} $PI_m:= [q_l , q_r]$. 
\end{algorithmic}
\end{algorithm}

Secondly, we can restrict the PI to be centered around some meaningful point. The symmetric PI centered around the estimated conditional mean $\widehat{H}_{0}(X_f)$ based on $\widehat{F}_{\text{Avg}}$ is developed in \cref{Algo:CaliPIsa}. We call such a prediction interval $PI_{sa}$ since it is symmetric and based on the average correction approach C3. Analogously, we consider a symmetric cPI but two endpoints are determined by $\widehat{F}_{\text{LtoR}}$ and $\widehat{F}_{\text{RtoL}}$, respectively; see \cref{Algo:CaliPIst} for details. Here, $\widehat{F}_{\text{LtoR}}$ and $\widehat{F}_{\text{RtoL}}$ are corrected CDF estimators based on approach C1 and C2 in \cref{Subsec:CorrectionforMono}, respectively. This special manipulation is intended to enlarge the length of the cPI, namely $PI_{st}$. 

In addition to this symmetric structure, we propose an asymmetric PI with $\widehat{F}_{\text{Avg}}$ in \cref{Algo:CaliPIaa}. We call such a prediction interval $P_{aa}$ since it is asymmetric and based on the average correction approach C3. Similarly to the extension from $PI_{sa}$ to $PI_{st}$, we can deploy $\widehat{F}_{\text{RtoL}}$ and $\widehat{F}_{\text{LtoR}}$ to enlarge $PI_{aa}$ to obtain a variant, namely $PI_{at}$. This procedure is described in \cref{Algo:CaliPIat}.  

\begin{Remark}
In these algorithms, we set up optimization problems to find indices $l$ and $r$ so that the endpoints of PI can be determined. For example, in \cref{Algo:CaliPIaa} and \cref{Algo:CaliPIat}, two indices should satisfy the inequalities $\widetilde{F}_{\text{Avg}}(q_{l}) \leq \alpha/2$,  $\widetilde{F}_{\text{Avg}}(q_{r}) \geq 1 - \alpha/2$ or $\widetilde{F}_{\text{LtoR}}(q_{l}) \leq \alpha/2$, $\widetilde{F}_{\text{RtoL}}(q_{r}) \geq 1- \alpha/2$. The less than $\leq$ and greater than $\geq$ used in two inequalities imply that the cPI is expected to have a coverage rate larger than $1-\alpha$ asymptotically. In practice with finite samples, this alleviates the undercoverage issue. Additionally, we did simulations with mean-centered and equal-tail cPI. However, the resulting PI length is so large compared to other cPIs, so we did not present the algorithm for such a type of PI here. 
\end{Remark}

\begin{algorithm}[htbp]
\caption{Symmetric calibration PI centered around mean with $\widetilde{F}_{\text{Avg}}$ and DNN: $PI_{sa}$}\label{Algo:CaliPIsa}
\setstretch{1.25}
\begin{algorithmic}[1]
\Require All required arguments in \cref{Algo:CaliPIwithDNN}.
\State Initialize $g+1$ number of DNN $H_{\theta,j}(X)\in \mathcal{F}_{\text{DNN}}$ for $j = 1,\cdots, g$; $\mathcal{F}_{\text{DNN}}$ is a DNN class with some structure constraint. 

\State Perform Steps 2 - 16 in \cref{Algo:CaliPIwithDNN} to get trained $\widehat{H}_{0},\widehat{H}_{1},\ldots,\widehat{H}_{g}$. Then, compute $\widehat{H}_{1}(X_f),\ldots,\widehat{H}_{g}(X_f)$ for a future $X_f$ of interest and take correction approach C3 to get $\widetilde{F}_{\text{Avg}}$. Also, compute $\widehat{H}_{0}(X_f)$ as estimation of the conditional mean of $Y$.
\State Fine the index $c\in\{1,\ldots,g\}$ s.t. $c = \arg\min_{c\in\{1,\ldots,g\}}|q_c - \widehat{H}_{0}(X_f)|$. 
\State Solve the below optimization problem to determine two indices $l$ and $r$:
\begin{equation*}
        \begin{split}
            & k = \arg\min_{k\in \mathbb{N}}\widetilde{F}_{\text{Avg}}(q_{r}) - \widetilde{F}_{\text{Avg}}(q_{l}) \geq 1-\alpha;\\
            &\text{where}~ r = \min(c+k,g)~\text{and}~ l = \max(c-k, 1).\\
        \end{split}
    \end{equation*}

\State \textbf{Return} $PI_{sa}:= [q_l , q_r]$. 
\end{algorithmic}
\end{algorithm}


\begin{algorithm}[htbp]
\caption{Symmetric calibration PI with $\widetilde{F}_{\text{RtoL}}$, $\widetilde{F}_{\text{LtoR}}$ and DNN: $PI_{st}$}\label{Algo:CaliPIst}
\setstretch{1.25}
\begin{algorithmic}[1]
\Require All required arguments in \cref{Algo:CaliPIwithDNN}.
\State Perform Steps 1-3 in \cref{Algo:CaliPIsa}.
\State Solve the below optimization problem to determine two indices $l$ and $r$:
\begin{equation*}
        \begin{split}
            & k = \arg\min_{k\in \mathbb{N}}\widetilde{F}_{\text{RtoL}}(q_{r}) - \widetilde{F}_{\text{LtoR}}(q_{l}) \geq 1-\alpha;\\
            &\text{where}~ r = \min(c+k,g)~\text{and}~ l = \max(c-k, 1).\\
        \end{split}
    \end{equation*}

\State \textbf{Return} $PI_{st}:= [q_l , q_r]$. 
\end{algorithmic}
\end{algorithm}


\begin{algorithm}[htbp]
\caption{Asymmetric calibration PI with $\widetilde{F}_{\text{Avg}}$ and DNN: $PI_{aa}$}\label{Algo:CaliPIaa}
\setstretch{1.25}
\begin{algorithmic}[1]
\Require All required arguments in \cref{Algo:CaliPIwithDNN}.
\State Initialize $g$ number of DNN $H_{\theta,j}(X)\in \mathcal{F}_{\text{DNN}}$ for $j = 1,\cdots, g$; $\mathcal{F}_{\text{DNN}}$ is a DNN class with some structure constraint. 

\State Perform Steps 7 - 16 in \cref{Algo:CaliPIwithDNN} to get trained $\widehat{H}_{1},\ldots,\widehat{H}_{g}$. Then, compute $\widehat{H}_{1}(X_f),\ldots,\widehat{H}_{g}(X_f)$ for a future $X_f$ of interest and take correction approach C3 to get $\widetilde{F}_{\text{Avg}}$.
\State Solve the two optimization problems to determine two indices $l$ and $r$:
    \begin{equation*}
        \begin{split}
            & l = \arg\max_{ 1 \leq k \leq g}\widetilde{F}_{\text{Avg}}(q_{k}) \leq \alpha/2;\\
            & r = \arg\min_{ 1 \leq k \leq g}\widetilde{F}_{\text{Avg}}(q_{k}) \geq 1 - \alpha/2.\\ 
        \end{split}
    \end{equation*}

\State \textbf{Return} $PI_{aa}:= [q_l , q_r]$. 
\end{algorithmic}
\end{algorithm}


\begin{algorithm}[htbp]
\caption{Asymmetric calibration PI with $\widetilde{F}_{\text{RtoL}}$, $\widetilde{F}_{\text{LtoR}}$ and DNN: $PI_{at}$}\label{Algo:CaliPIat}
\setstretch{1.25}
\begin{algorithmic}[1]
\Require All required arguments in \cref{Algo:CaliPIwithDNN}.
\State Perform Steps 1 - 3 in \cref{Algo:CaliPIaa}.
\State Solve the two optimization problems to determine two indices $l$ and $r$:
    \begin{equation*}
        \begin{split}
            & l = \arg\max_{ 1 \leq k \leq g}\widetilde{F}_{\text{LtoR}}(q_{k}) \leq \alpha/2;\\
            & r = \arg\min_{ 1 \leq k \leq g}\widetilde{F}_{\text{RtoL}}(q_{k}) \geq 1 - \alpha/2.\\ 
        \end{split}
    \end{equation*}

\State \textbf{Return} $PI_{at}:= [q_l , q_r]$. 
\end{algorithmic}
\end{algorithm}

\FloatBarrier
Here, we provide a specific illustration to build a calibration PI with the kernel estimator based on the corrected CDF estimator $\widetilde{F}_{\text{Avg}}$. We call such cPI $PI_{aak}$ since it is asymmetric and based on the average correction approach with a kernel estimator. Other types of cPI with the kernel estimator can be built similarly to \cref{Algo:CaliPIm,Algo:CaliPIsa,Algo:CaliPIst,Algo:CaliPIat}.

\begin{algorithm}[htbp]
\caption{Asymmetric cPI with $\widetilde{F}_{\text{Avg}}$ and kernel estimators: $PI_{aak}$}\label{Algo:CaliPIaak}
\setstretch{1.25}
\begin{algorithmic}[1]
\Require All required arguments in \cref{Algo:CaliPIwithKernel}.
\State Apply the bandwidth-choice strategy to determine the kernel estimators $\hat{f}(x,y)$ and $\hat{f}(x)$. Then, the conditional density estimator $\hat{f}(y|x)$ is obtained. 

\State Perform Steps 2 - 3  in \cref{Algo:CaliPIwithKernel} to get $\widetilde{F}$, which is an intrinsically appropriate estimation of the conditional CDF of $Y$ given $X_f$.
\State Solve the two optimization problems to determine two indices $l$ and $r$:
    \begin{equation*}
        \begin{split}
            & l = \arg\max_{ 1 \leq k \leq g}\widetilde{F}_{\text{Avg}}(q_{k}) \leq \alpha/2;\\
            & r = \arg\min_{ 1 \leq k \leq g}\widetilde{F}_{\text{Avg}}(q_{k}) \geq 1 - \alpha/2.\\ 
        \end{split}
    \end{equation*}

\State \textbf{Return} $PI_{aak}:= [q_l , q_r]$. 
\end{algorithmic}
\end{algorithm}

\section{Theoretical Analysis}\label{Sec:Theorey}
We justify the asymptotic validity of the cPI with DNN and kernel estimators. Moreover, we explore the effects of the equal distance $D$ between grid points on the coverage rate. We show the possibility that a cPI guarantees at least $1-\alpha$ coverage rate even in the finite sample situations if all involved DNNs can be estimated in an optimal way (oracle estimation) and the grid distance $D$ satisfies mild conditions. We claim this finite sample coverage guarantee is the strongest property, which is hardly achievable in practice. Thus, we aim for a cPI that guarantees at least $1-\alpha$ coverage rate with a high probability when the sample size is large enough. As mentioned in \cref{Sec:Intro}, this high-probability coverage warranty is finer than the asymptotic validity. We show that the cPI with kernel estimators exhibits such a feature. For cPI with DNN, we conjecture this property still holds since it can be thought of as an intermediate stage between the finite sample and asymptotic coverage guarantee. The properties of cPI with DNN- or kernel estimators are summarized in \cref{Tab:perpertiesofCPI}. We first give the theory related to cPI with DNN and then discuss the property of cPI with kernel estimators.

\begin{table}[htbp]
    \centering
     \caption{The comparison of calibration PI based on DNN and kernel estimators.}
    \begin{tabular}{lcc}
    \toprule
     & DNN-based &  kernel-based \\
     \midrule
      Finite-sample coverage (oracle estimation)  &  \selfcheckmark  & ---  \\
      Large-sample coverage (high probability)  &  \notcheckmark &  \selfcheckmark  \\
      Asymptotic coverage    &  \selfcheckmark  &  \selfcheckmark \\
      \bottomrule
      \label{Tab:perpertiesofCPI}
    \end{tabular}
    \vspace{2pt}
    
    \raggedright
    {\footnotesize Note: We use the symbol \notcheckmark~ to remark that we do not give rigorous proof to show the high-probability large-sample coverage property of DNN-based cPI, but we conjecture it is true. The simulation and empirical studies also favor the cPI with DNN estimators.}
\end{table}

\subsection{Theoretical Property of Calibration PI with DNN}

\subsubsection{Asymptotic coverage}
All cPIs considered so far demand the estimation of the CDF as a function of $X$ when it is evaluated on some fixed grid points, and the conditional mean if the cPI is required to be symmetric around the mean. As revealed in \cref{Lemma:CDE} and \cref{Lemma:EstHk}, under minimal assumptions A1 to A3, the oracle conditional mean $H_0(X)$ and CDF functions $H_j(X)$ are all continuous w.r.t. $X$, which underscores the potential to estimate them by DNNs due to their universal approximation property. 

In this paper, we consider the standard fully connected forward DNN with ReLU activation functions. It can be treated as a ``black box'' that takes the input $X$ and returns output in the following way:
$$
    H_{\theta}(X) =  \varpi_{L}(\sigma(\varpi_{L-1}(\cdots \sigma(\varpi_{2}\sigma(\varpi_{1}\sigma(\varpi_0X + b_{0}) +b_1)  + b_2)\cdots) +  b_{L-1}) + b_{L};
$$
where $L$ represents the depth of the hidden layer; $\sigma(\cdot)$ is the activation function; $\{b_0,\ldots,b_L\}$ and $\{\varpi_0,\ldots,\varpi_L\}$ are bias vectors and weight matrices, respectively; they stand for all parameters $\theta$ in a DNN and will be trained with some gradient descent algorithm, e.g., Adam. Recently, the error bound of the least square DNN estimator on estimating regression functions has been studied relying on different settings. To relieve the curse of dimensionality \cite{jiao2023deep, nakada2020adaptive, schmidt2019deep} made distributional assumptions on the independent variables $X$, e.g., $X$ is supported on a low-dimensional manifold; \cite{schmidt2020nonparametric, bauer2019deep} assumed the underlying regression has some special structure.  Without any regression model and distribution requirement, \cite{farrell2021deep} showed the $L^2$ error of a DNN estimator and the true function converges to 0 with a high probability asymptotically, such as $\|\widehat{H}_0 - H_0\|_{L^2(X)} \to 0$ as $n\to\infty$ with probability tends to 1 in our context; here the expectation is only taken w.r.t. $X$. Furthermore, \cite{wu2024scalable} improved the error bound of DNN's estimation ability for functions belonging to an H\"{o}lder space with a so-called scalable subsampling technique proposed by \cite{politis2024scalable}. 

In this paper, we do not make additional assumptions on the distribution of $X$ and the structure of the conditional mean and the conditional CDF. However, to show the consistent result of the DNN and kernel estimators, we make three additional assumptions to shape the joint distribution $P_{X,Y}$ and the marginal density of $X$ for the theoretical proof:
\begin{itemize}
    \item B1 The domain of $Y$ and $X$ are compact sets, respectively, i.e., $\mathcal{Y} := [-M_1,M_1]$ and $\mathcal{X} := [-M_2,M_2]^d$; $M_1$ and $M_2$ are two arbitrarily large constants;
    \item B2 The density function of $X$, namely $f(x)$ is positive and bounded on the whole domain $\mathcal{X}$; the conditional density $f(y|x)$ is also bounded for all $x$;
    \item B3 The marginal density $f(x)$ and joint density $f(x,y)$ are twice differentiable. The matrix norm of the Hessian matrix of $f(x)$ and $f(x,y)$ is finite for all $x$ and $y$, e.g., $\|\nabla^2f(x)\|_1$ and $\|\nabla^2f(x,y)\|_1$ are finite. 
\end{itemize}
Accordingly, the error terms could be truncated to have a compact domain to simplify the proof. In simulation studies, we do not perform this truncation. Subsequently, we develop the asymptotic coverage guarantee for our cPIs.

\begin{Remark}
    We should mention that only assumptions B1 and B2 are required to give proof related to cPI with DNN estimators. To develop the theory about cPI with kernel methods, we need a stronger assumption of smoothness (i.e., B3) than the requirement of cPI with DNN.
\end{Remark}

\begin{Theorem}\label{Theorem:coverageasy}
Under A1 to A3 and B1 to B2, let the DNN class $\mathcal{F}_{\text{DNN}}$ contain all DNN which have a width $W := 3^{d+3} \max \left\{d\left\lfloor N_1^{1 / d}\right\rfloor, N_1+1\right\}$ and depth $L := 12 N_2+14+2d$; $N_1 = \ceil{ \frac{n^{\frac{d}{2(\tau+d)}}}{\log n}}$ and $N_2 = \ceil{ \log(n) }$; $\tau >0$; see \cref{Remakr:explanationsofWL} for the meaning of $W$ and $L$ in our context. Also, we require the DNN in $\mathcal{F}_{\text{DNN}}$ is truncated, i.e., $\{|f_{\text{DNN}}|\leq M_1: f_{\text{DNN}} \in \mathcal{F}_{\text{DNN}}\}$. Then, for large enough $n$, i.e., $n> \max((2eM_1)^2, \text{Pdim}(\mathcal{F}_{\text{DNN}}))$, and samples $S_n:=\{(X_i,Y_i)\}_{i=1}^n $ s.t., $\mathbb{P}(S_n\notin A_n) = o(1)$, where $A_n\subseteq\mathcal{X}\times\mathcal{Y}$ is an expanding set as $n$ increasing, we have
$$
\mathbb{P}(Y\in \widehat{\mathcal{I}} |X_f = x_f) \geq 1-\alpha
$$
for any $x_f$ in its domain $\mathcal{X}$ as $n\to\infty$; where $\widehat{\mathcal{I}}$ represents all cPIs computed according to \cref{Algo:CaliPIm,Algo:CaliPIsa,Algo:CaliPIst,Algo:CaliPIaa,Algo:CaliPIat}; $\text{Pdim}(\mathcal{F}_{\text{DNN}})$ is the pseudo-dimension of DNN class $\mathcal{F}_{\text{DNN}}$; see \cite{bartlett2019nearly} for the formal definition. 
\end{Theorem}

\begin{Remark}[The width and depth of the DNN in \cref{Theorem:coverageasy}]\label{Remakr:explanationsofWL}
The width $W$ and depth $L$ mean the number of hidden layers of a DNN in $\mathcal{F}_{\text{DNN}}$ is no more than $L$ and the maximum number of neurons among all hidden layers is no more than $W$, respectively. With appropriate approximation results, all theories in this paper can be extended to a general DNN in which all neurons may not be fully and feedforward connected. 

\end{Remark}

 \subsubsection{Finite-sample coverage}\label{SubsubSec:finietcover}
 Notably, in the finite-sample scenarios, the Lemma 1 of \cite{lei2014distribution} showed that it is impossible to make a PI with a finite length that guarantees at least $1-\alpha$ coverage for $Y$ conditional on $X_f = x_f$ for any distribution of $P_{X,Y}$ and all $x_f\in N(P)$, which is the so-called non-atoms\footnote{ Let $P_X$ denote the marginal distribution of $X$ under $P_{X,Y}$. A point $x$ is a non-atom for $P_X$ if $x$ is in the support of $P_X$ and if $P_X\{B(x, \delta)\} \rightarrow 0$ as $\delta \rightarrow 0$, where $B(x, \delta)$ is the Euclidean ball centred at $x$ with radius $\delta$. $N(P)$ is the set of all non-atom points.}. 

Here, we show that our cPI can guarantee the finite-sample coverage property under several conditions: (1) We do not consider the conditional coverage validity for any distribution $P_{X,Y}$; we restrict our interests on joint distribution satisfying A1 to A3 and B1 to B2; (2) We assume the DNN can be trained in an optimal way, i.e., the $\widehat{H}_{j}$ can achieve the optimal performance among all DNN in $\mathcal{F}_{\text{DNN}}$; (3) We make a mild assumption on the equal gap distance $D$ of grid points and further adjust the cPI to warrant the finite-sample coverage.


The second condition is devoted to remove the ``gap'' between $\widehat{H}_j$ which minimizes the empirical risk and the optimal DNN in $\mathcal{F}_{\text{DNN}}$ to approximate desired functions, namely $H^*_j : = \arg\min_{H_\theta\in\mathcal{F}_{\text{DNN}}}\|H_{\theta,j} - H_j \|_{L^\infty}$. We should notice that the performance of $H_j^*$ only depends on the structure of the DNN class $\mathcal{F}_{\text{DNN}}$, i.e., the number of neurons, width, and depth. We assume that the $\widehat{H}_{j}$ can be trained as well as the $H^*_{j}$. This condition is formulated by the assumption and the remark below:
\begin{itemize}
    \item B4: The $L^\infty$ norm of the difference between $\widehat{H}_j$ and $H^*_j$ is bounded by $\delta_n$, i.e., $\| \widehat{H}_j - H^*_j  \|_{L^\infty(X)}\leq \delta_n$, with an appropriate sequence $\delta_n\to0$ as $n\to\infty$. 
\end{itemize}
\begin{Remark}\label{Remakr:2ndcondition}
    As far as we know, there is no existing literature regarding the asymptotic result required in B4. Once such a result is proved in the future, we can get a precise understanding of the grid distance's effects on finite-sample coverage. To illustrate it a little bit further, suppose $\|H_{j} - H^*_j \|_{L^\infty} = \zeta_{W,L}$; $\zeta_{W,L}$ is the error that depends only on the width and depth of a DNN. Then, we have $\| \widehat{H}_j - H_{j}  \|_{L^\infty(X)}\leq \delta_n + \zeta_{W,L}$. If the grid point distance $D$ is large enough so that the true CDF value difference evaluated at any two adjacent grid points is larger than the $L^\infty$ norm error between $\widehat{H}_j$ and $H_{j}$, cPI can be adjusted to guarantee the coverage rate even in the finite-sample cases. For this moment, we treat the error sequence $\delta_n$ to be 0 for all $n$. This is possible but may not be practically accomplishable.   
\end{Remark}

For the last condition, we take the adjustment of $PI_{aa}$ as an example. We can move the determined left index $l$ and right index $r$ of $PI_{aa}$ one step more left and right, respectively, resulting in a so-called $PI_{aaa}$; we call such a PI the adjusted version of cPI; see \cref{Algo:CaliPIaaadjusted} below for details. 

\begin{algorithm}\label{Algo:PIaaa}
\caption{Adjusted Asymmetric calibration PI with $\widetilde{F}_{\text{Avg}}$: $PI_{aaa}$}\label{Algo:CaliPIaaadjusted}
\setstretch{1.25}
\begin{algorithmic}[1]
\Require All required arguments in \cref{Algo:CaliPIwithDNN}.
\State Initialize $g$ number of DNN $H_{\theta,j}(X)\in \mathcal{F}_{\text{DNN}}$ for $j = 1,\cdots, g$; $\mathcal{F}_{\text{DNN}}$ is a DNN class with some structure constraint. 
\State Perform Steps 2-3 in \cref{Algo:CaliPIaa}.
\State Adjust the current left index $l$ and right index $r$ by defining $\tilde{l} = max(1, l-1)$ and $\tilde{r} = min(g,r+1)$.
\State \textbf{Return} $PI_{aaa}:= [q_{\tilde{l}} , q_{\tilde{r}}]$. 
\end{algorithmic}
\end{algorithm}

Lastly, under the assumption about the equal grid distance $D$ between grid points, i.e., 
\begin{itemize}
    \item B5: The equal grid point distance $D$ is large enough s.t., $H_{i} - H_{j}\geq C\cdot  \max_{k\in\{1,\ldots,g\}}$ $\left(\omega_{H_k}^{\mathcal{X}}\left(2M_2 N_1^{-2 / d} N_2^{-2 / d}\right)\right)$ for some constant $C$ and any $i, j \in\{1,\ldots,g\}$ and $i = j + 1$; $\omega_{H_k}^{\mathcal{X}}(r)$ is so-called modulus of continuity of the target function $H_k$ defined on $\mathcal{X}$ for $k\in\{1,\ldots,g\}$; see \cref{Lemma:DNNC0f} for the formal definition.
\end{itemize}
We can give the theoretical reasoning for the finite-sample coverage of $PI_{aaa}$ in \cref{Theorem:finiteCoverage}.
\begin{Theorem}\label{Theorem:finiteCoverage}
    Under A1 to A3, B1 to B2, B4 to B5, and the condition described in \cref{Remakr:2ndcondition}, if the DNN structure is complex enough,  s.t., the $L^\infty$ error between $H_j^*$ and $H_j$ is less than 1 (this is for performing calibration possibly, since $H_i(x_f) - H_j(x_f) \leq 1$ for any $x_f$; we need the maximum $L^\infty$ error between $H_j^*(x_f)$ and $H_j(x_f)$ to be less than the distance between adjacent $H_i$ and $H_{i+1}$ so that the cPI can compensate for the estimation variability), then the $PI_{aaa}$ covers $Y$ conditional on any $x_f \in \mathcal{X}$ with at least probability $1-\alpha$, i.e., 
$$
\mathbb{P}(Y\in PI_{aaa} |X_f = x_f) \geq 1-\alpha.
$$
\end{Theorem}

\begin{Remark}[Further understanding of assumption B5]
The direct implication of B5 is that we need an equal distance $D$ between grid points to be large and then make the adjustment described in \cref{Algo:PIaaa} to compensate for the uncaptured estimation error of $H_j^*$ for finite sample cases to ensure coverage rate. In other words, with $PI_{aaa}$, the coverage rate shall increase as we decrease the number of grid points $g$ intuitively. However, in practice, this phenomenon may have a more complicated pattern. In the empirical study, we find that the coverage rate of $PI_{aaa}$ may decrease a little bit first and then increase eventually as the number of grid points keeps decreasing. In addition, the $PI_{aaa}$ could be $[\min(y_1,\ldots,y_n), \max(y_1,\ldots,y_n) ]$ when the DNN structure is too simple so that the estimation error is too large; $(y_1,\ldots,y_n)$ are observations of $(Y_1, \ldots, Y_n)$. 

\end{Remark}

\subsection{Theoretical Property of cPI with Kernel Estimators}
With the conditional density estimator $\hat{f}(y|x)$ defined by expressions in \cref{Eq:Kernelestimators}, we first show a non-asymptotic error bound of $\hat{f}(y|x)$ on estimating the true conditional density $f(y|x)$ under standard assumptions in \cref{Lemma:Errorboundofkernel}. The asymptotic validity of cPIs with kernel estimators, e.g., $PI_{aak}$ defined in \cref{Algo:CaliPIaak}, is implied by \cref{Lemma:Errorboundofkernel} directly. Moreover, we can show that the kernel-based cPI guarantees the coverage rate with a high probability when the sample size is larger than some constant. 

\subsubsection{Asymptotic coverage}
To develop the theoretical basis of cPI with kernel estimators, we impose some additional requirements on the kernel function $k(\cdot)$ as follows:
\begin{itemize}
    \item D1: $k(\cdot)$ is positive, differentiable, and bounded on its whole domain, and the derivative $k^{'}(\cdot)$ is also bounded; 
    \item D2: $\int v^2k(v) dv$ is finite. 
\end{itemize}
\begin{Remark}
    A direct result of D1 and the application of product kernel is that $K$ and $\widetilde{K}$ are Lipschitz continuous. This property will be utilized in the proof later. 
\end{Remark}

We verify the asymptotic validity of cPI with kernel estimators by developing a non-asymptotic error bound on density estimation as follows:
\begin{Lemma}\label{Lemma:Errorboundofkernel}
Under B1 to B3, D1 to D2, for the conditional density estimator defined by \cref{Eq:Kernelestimators}, when the sample size is large enough, s.t., $n> \max(N_3,N_4)$, where,
$$
N_3~\text{is a large constant s.t.,}~ 2A_1\frac{\ln(N_3)}{N_3h^{d+1}}^{1/2} \leq \frac{1}{2}~;~N_4~\text{is a large constant s.t.,}~ 2A^{\prime}_1\frac{\ln(N_4)}{N_4h^{d}}^{1/2} \leq \frac{1}{2}; 
$$
where $A_1 = \sup_z\widetilde{K}(z)$ and $A^{\prime}_1 = \sup_xK(x)$. Then, we have 
\begin{equation}\label{Eq:kernelesterror}
    \sup_{x,y}|\hat{f}(y|x) - f(y|x)| = O\left( C\left(\frac{\ln(n)}{nh^{d+1}}\right)^{1/2} + C^{\prime} \left(\frac{\ln(n)}{nh^{d}}\right)^{1/2} + h^2\right),
\end{equation}
with probability $\kappa_n \to 1$ in some rate as $n\to \infty$, $h\to 0$, and $\ln(n)/(nh^{d+1})\to\infty$; $C$ and $C^{\prime}$ are two appropriate large constants required in the proof. We can simplify the term on the r.h.s. of \cref{Eq:kernelesterror} as $O\left( \left(\frac{\ln(n)}{nh^{d+1}}\right)^{1/2}  + h^2\right)$.
\end{Lemma}

A direct implication of \cref{Lemma:Errorboundofkernel} is that $\hat{f}(y|x)$ is consistent with $f(y|x)$ uniformly. Thus, $\widehat{H}^k_{0}(X_f)$ and $\widehat{H}^k_{1}(X_f), \ldots, \widehat{H}^k_{g}(X_f)$ defined in \cref{Algo:CaliPIwithKernel} are consistent with true conditional mean $H_0$ and true conditional CDF evaluated at points $q_1,\ldots,q_g$, i.e., $\mathbb{E}(\mathbbm{1}(Y\leq q_j)|X)$ for $j = 1,\ldots, g$, respectively. As a result, all cPIs with DNN estimators guarantee at least $1-\alpha$ conditional coverage asymptotically. 

\subsubsection{Large-sample coverage}
Beyond the asymptotic coverage property, we can show that there is a high probability that the cPI, which is based on kernel estimators, guarantees the coverage rate when the sample size is large. Compared to the finite-sample coverage property explained in \cref{SubsubSec:finietcover} where the oracle estimation of DNN is required, the large-sample coverage is more achievable. We focus on the adjusted version of the asymmetric cPI described in \cref{Algo:CaliPIaak}, i.e., adjusting the returned left index $l$ and right index $r$ by defining $\tilde{l} = max(1, l-1)$ and $\tilde{r} = min(g,r+1)$ as what we did in \cref{Algo:PIaaa}; we call the adjusted asymmetric cPI based on kernel $PI_{aaak}$. Under the condition about the equal distance of grid points, i.e., 
\begin{itemize}
    \item B$^\prime$5: The equal distance $D$ between grid points $q_j$ and $q_{j+1}$ is large enough s.t., $H_{i} - H_{j}\geq C\cdot \left( \left(\frac{\ln(n)}{nh^{d+1}}\right)^{1/2} + \left(\frac{\ln(n)}{nh^{d}}\right)^{1/2} \right. $ $\left. + h^2\right)$ for some constant $C$ and any $i, j \in\{1,\ldots,g\}$ and $i  = j+1$; this is the condition to perform calibration for compensating the estimation variability in the prediction procedure; review \cref{Remakr:2ndcondition} for more thoughts.
\end{itemize}
In addition, suppose that there is a constant $N_5$ s.t., $\left|\int^{q_j}_{-\infty} \left( \hat{f}(y|x) - f(y|x) \right) dy \right| < 1$ for all $j = 1,\ldots, g$ and all $x\in\mathcal{X}$ when $n>N_5$. Then, we have:
\begin{Theorem}\label{Theorem:largePropertyCPIK}
    Under B1 to B3, B$^{\prime}$5, D1 to D2, with the sample size $n>\max(N3,N4,N5)$, where $N_3$ and $N_4$ are defined in \cref{Lemma:Errorboundofkernel}, for samples $S_n:=\{(X_i,Y_i)\}_{i=1}^n $ s.t., $\mathbb{P}(S_n\notin A_n) = o(1)$, where $A_n\subseteq\mathcal{X}\times\mathcal{Y}$ is an expanding set as $n$ increasing, we can conclude that:
    $$
\mathbb{P}(Y\in PI_{aaak} |X_f = x_f) \geq 1-\alpha
$$
for any $x_f\in\mathcal{X}$ as $n\to \infty$, $h\to 0$, and $\ln(n)/(nh^{d+1})\to\infty$. 
\end{Theorem}
The proof of \cref{Theorem:largePropertyCPIK} is straightforward and similar to the proof of \cref{Theorem:finiteCoverage}.

\section{Simulations}\label{Sec:Simulation}
We first deploy simulation studies to verify the performance of cPI based on DNN estimators. Due to the issue of limited computational resources, we leave the validation of cPI with kernel estimators to \cref{Sec:RealData} where several PIs based on deep-generative approaches are also compared with our various cPIs for real data analysis. In total, 6 non-linear models in homoscedastic and heteroscedastic formats with different error distributions are considered:
\begin{itemize}
    \item Model-1: $Y_i = X_{i,1}^2 + \sin(X_{i,2} + X_{i,3}) + \epsilon_{i,z}$;
    \item Model-2: $Y_i = X_{i,1}^2 + \sin(X_{i,2} + X_{i,3}) + \epsilon_{i,t}$;
    \item Model-3: $Y_i = X_{i,1}^2 + \sin(X_{i,2} + X_{i,3}) + \epsilon_{i,s}$;
    \item Model-4: $Y_i=X_{i,1}^2+\exp \left(X_{i,2}+X_{i,3} / 3\right)+$ $X_{i,4}-X_{i,5}+\left(0.5+X_{i,2}^2 / 2+X_{i,5}^2 / 2\right) \cdot \epsilon_{i,z}$; 
    \item Model-5: $Y_i=X_{i,1}^2+\exp \left(X_{i,2}+X_{i,3} / 3\right)+$ $X_{i,4}-X_{i,5}+\left(0.5+X_{i,2}^2 / 2+X_{i,5}^2 / 2\right) \cdot \epsilon_{i,t}$;
    \item Model-6: $Y_i=X_{i,1}^2+\exp \left(X_{i,2}+X_{i,3} / 3\right)+$ $X_{i,4}-X_{i,5}+\left(0.5+X_{i,2}^2 / 2+X_{i,5}^2 / 2\right) \cdot \epsilon_{i,s}$;
\end{itemize}
$\epsilon_{i,z} \sim N(0,1)$; $\epsilon_{i,t} \sim t_5$; $\epsilon_{i,s} \sim \text{Skew-normal}(0,1,10)$; $(X_{i,1},X_{i,2},X_{i,3},X_{i,4},X_{i,5})\sim N(0,I_5)$ for all $i$. When we apply the skew-normal error distribution, we normalize the variance of the generated sample to 1. 

To verify the asymptotic property of cPI to some extent, we consider two relatively large sample sizes $ n = 2000$ and $n = 10000$. In our expectation, the performance of cPI improves as the sample size increases. Besides, the naive PI (\ref{Eq:PInormal}), taken as the benchmark, still returns a coverage rate less than the nominal level when the sample size is large. This is not in surprise since the naive PI (\ref{Eq:PInormal}) heavily relies on the normality assumption and does not attempt to capture the estimation variability. We denote PI (\ref{Eq:PInormal}) by $PI_{b}$.

For cPI, we consider five cPIs with different calibration approaches described in \cref{Sec:DiffCali}. To compare different methods, we generate $\{Y_i,X_i\}_{i=1}^{n}$ as the training sample; $\{X_i\}_{i=1}^n$ are vectors of predictors. For all DNNs involved in different PIs, we use the same structure, i.e., $[W,W]$ which stands for a two-layer DNN and each layer has $W$ neurons. For the hyperparameter setting, we take $lr = 0.001$; $B = 200$; $E = 2000$; $g = 200$; $m = 20$; $\alpha = 0.05$ and apply \textit{Adam} optimizer. We also consider five different widths of DNN, i.e., $W = 10, 20, \ldots, 50$. To make our simulation studies reproducible, we do replications $S=500$ times and generate training samples with $seed(i)$ for $i = 1,\ldots, 500$. To test the performance of prediction intervals for all replications, we consider multiple interests of $X^{f}$ and corresponding $Y^f$. We generate test data $\{Y^f_{j},X^f_{j}\}_{j=1}^{T}$ with $seed(1000)$ and evaluate the empirical coverage rate on $\{Y^f_{j}\}_{j=1}^T$ given $X^f_{j}$.

\begin{Remark}
    For the PI (\ref{Eq:PInormal}), it is not well-defined if $\hat{\kappa}_{Y|X^f_{j}} -  \hat{\mu}^2_{Y|X^f_{j}}<0$ for the $j$-th test point. We ignore all these test points in the evaluation process. 
\end{Remark}
To make evaluations, we compute the average coverage rate (CR) and length (AL) from 500 replications over $T = 2000$ test points as follows:
\begin{equation}\label{Eq:CRAL}
    \text{CR}:= \frac{1}{S}\sum_{i = 1}^{S}\frac{1}{\widetilde{T}}\sum_{j=1}^{\widetilde{T}}\mathbb{P}( Y^{f}_{i,j}\in \widehat{I}_C|X^{f}_{i,j})~;~\text{AL}:= \frac{1}{S}\sum_{i = 1}^{S}\frac{1}{\widetilde{T}}\sum_{j=1}^{\widetilde{T}} (R_{i,j} - L_{i,j});
\end{equation}
where $S$ is the number of replications, which is 500 in our simulation studies; $\widetilde{T}$ is the number of test points out of total $T$ test points whose corresponding $PI_{b}$ is well-defined; $Y^{f}_{i,j}$ is the future response value for the $j$-th test point in the $i$-th replication; $\widehat{I}_C$ represents different PIs conditional on $X^{f}_{i,j}$; $R_{i,j}$ and $L_{i,j}$ are the right and left endpoints of different PIs for the $j$-th test data in the $i$-th replication. Since the simulation model is known to us, we can approximate the conditional probability $\mathbb{P}( Y^{f}_{i,j}\in \widehat{I}_C|X^{f}_{i,j})$ by the Monte Carlo simulation technique, i.e., we define 
\begin{equation}
    \widehat{\mathbb{P}}( Y^{f}_{i,j}\in \widehat{I}_C|X^{f}_{i,j}) = \frac{1}{V}\sum_{v = 1}^V \mathbbm{1}( Y^{f}_{i,j,v} \in \widehat{I}_C | X^{f}_{i,j} );
\end{equation}
where $Y^{f}_{i,j,v}$ is the $v$-th pseudo values simulated from each model conditional on $X^{f}_{i,j}$. We take $V = 5000$ in our simulation study. Then, we replace $\mathbb{P}( Y^{f}_{i,j}\in \widehat{I}_C|X^{f}_{i,j})$ in \cref{Eq:CRAL} by $ \widehat{\mathbb{P}}( Y^{f}_{i,j}\in \widehat{I}_C|X^{f}_{i,j})$ to compute the CR. Simulation results for Models 1-3 and Models 4-6 are presented in \cref{Tab:SimresM1-3,Tab:SimresM4-6}, respectively. We can summarize several findings:
\begin{itemize}
    \item The benchmark $PI_{b}$ undercovers the test points with an average coverage rate lower than the nominal level. This phenomenon sustains even when the sample size is large. One exception is for Model-6, where $PI_{b}$ overcovers the test points, but with a pretty large AL as a sacrifice.
    
    \item For our cPIs, $PI_{m}$ is the one with the minimum length, but the coverage rate is the lowest among the other cPIs. Although the CR of $PI_{m}$ is less than the nominal level when the sample size is small, it still has a better CR than the benchmark $PI_b$ and is shorter than $PI_b$. For the two pairs of symmetric and asymmetric PIs, the variants with two endpoints determined by $\widetilde{F}_{\text{RtoL}}$ and $\widetilde{F}_{\text{LtoR}}$ always have a larger AL than the corresponding one with endpoints regulated by $\widetilde{F}_{\text{Avg}}$, so the CR of $PI_{st}$ and $PI_{at}$ are higher than $PI_{sa}$ and $PI_{aa}$, respectively. This reason is that the application of $\widetilde{F}_{\text{RtoL}}$ and $\widetilde{F}_{\text{LtoR}}$ is designed to enlarge the resulting PI. Regarding the overall performance on CR and AL, $PI_{aa}$ works best; it can deliver a CR close to a nominal level even for small sample scenarios; simultaneously, it possesses even a shorter length compared to $PI_b$ for some scenarios.

    \item As the sample size increases, the behavior of all cPIs verifies the asymptotic consistency, i.e., their CR increases or AL decreases or these two changes happen together. In other words, as the sample size increases, the performance of cPIs gets better. On the other hand, the CR of $PI_b$ hardly changes since the undercoverage issue of $PI_b$ hinges on the loss of capturing the estimation variability and the model misspecification when the error distribution is non-normal. Notably, $PI_b$ can not give a higher CR than the nominal level, even for Models 1 and 4 in which the error distribution is normal. This highlights the necessity of capturing the estimation variability in PI. By design, our cPIs allow a systematic way to compensate for the estimation variability, which is not usually directly captured.
    
    \item When the sample size is small, the general performance of different PIs gets worse as the size of the DNN increases. According to the double descent error property of DNN estimation, we postulate that the performance of different PIs will get better and better if we keep increasing the size of the DNN. However, this manipulation will increase the usage of memory and extend the running time unnecessarily. 
\end{itemize}

Based on the simulation results presented in \cref{Tab:SimresM1-3,Tab:SimresM4-6}, we recommend $PI_{aa}$ as the optimal cPI when the training sample size is moderate. The short sample scenario will be considered with real data in the next section.
\begin{table}[htbp]
    \centering
\caption{Average coverage rate (CR) and length (AL) of different PIs for Models 1-3.}
\begin{tabular}{lllllll|llllll}

\toprule
 & $PI_b$ & $PI_{m}$ & $PI_{sa}$ & $PI_{st}$ & $PI_{aa}$ & $PI_{at}$ & $PI_b$ & $PI_{m}$ & $PI_{sa}$ & $PI_{st}$ & $PI_{aa}$ & $PI_{at}$ \\
\midrule
 & \multicolumn{6}{c}{$n = 2000$}   & \multicolumn{6}{c}{$n = 10000$}  \\
Model-1 &    & \\[2pt]
CR-w10 & 0.916 & 0.931 & 0.960 & 0.985 & \textbf{0.951} & 0.978 & 0.918 & 0.937 & 0.958 & 0.977 & \textbf{0.952} & 0.972 \\
AL-w10 & 4.267 & 4.174 & 4.712 & 5.888 & \textbf{4.664} & 5.906 & 4.155 & 3.962 & 4.296 & 5.055 & \textbf{4.277} & 5.120 \\
CR-w20 & 0.906 & 0.930 & 0.969 & 0.994 & \textbf{0.956} & 0.985 & 0.919 & 0.933 & 0.959 & 0.982 & \textbf{0.950} & 0.976 \\
AL-w20 & 4.115 & 4.479 & 5.221 & 6.823 & \textbf{5.155} & 6.911 & 3.985 & 4.023 & 4.452 & 5.562 & \textbf{4.428} & 5.667 \\
CR-w30 & 0.899 & 0.934 & 0.976 & 0.996 & \textbf{0.962} & 0.988 & 0.916 & 0.930 & 0.962 & 0.987 & \textbf{0.951} & 0.980 \\
AL-w30 & 4.117 & 4.771 & 5.648 & 7.357 & \textbf{5.563} & 7.437 & 3.975 & 4.120 & 4.660 & 5.958 & \textbf{4.613} & 6.089 \\
CR-w40 & 0.893 & 0.937 & 0.980 & 0.997 & \textbf{0.966} & 0.989 & 0.912 & 0.927 & 0.964 & 0.989 & \textbf{0.952} & 0.982 \\
AL-w40 & 4.123 & 4.963 & 5.911 & 7.594 & \textbf{5.805} & 7.647 & 3.979 & 4.195 & 4.820 & 6.141 & \textbf{4.754} & 6.261 \\
CR-w50 & 0.887 & 0.939 & 0.982 & 0.997 & \textbf{0.968} & 0.989 & 0.908 & 0.925 & 0.965 & 0.989 & \textbf{0.953} & 0.982 \\
AL-w50 & 4.151 & 5.089 & 6.081 & 7.694 & \textbf{5.952} & 7.699 & 3.986 & 4.238 & 4.916 & 6.183 & \textbf{4.833} & 6.278 \\[5pt]

Model-2 &    &  \\[2pt]
CR-w10 & 0.921 & 0.932 & 0.960 & 0.985 & \textbf{0.951} & 0.978 & 0.929 & 0.937 & 0.955 & 0.974 & \textbf{0.950} & 0.970 \\
AL-w10 & 5.271 & 5.555 & 6.490 & 8.625 & \textbf{6.342} & 8.368 & 5.168 & 5.172 & 5.705 & 7.055 & \textbf{5.668} & 6.974 \\
CR-w20 & 0.911 & 0.934 & 0.970 & 0.991 & \textbf{0.957} & 0.985 & 0.928 & 0.933 & \textbf{0.957} & 0.981 & 0.949 & 0.975 \\
AL-w20 & 5.170 & 6.112 & 7.354 & 9.986 & \textbf{7.113} & 9.696 & 5.092 & 5.312 & \textbf{6.040} & 8.029 & 5.948 & 7.877 \\
CR-w30 & 0.902 & 0.938 & 0.976 & 0.994 & \textbf{0.963} & 0.988 & 0.924 & 0.930 & \textbf{0.960} & 0.985 & 0.949 & 0.978 \\
AL-w30 & 5.177 & 6.549 & 7.943 & 10.68 & \textbf{7.645} & 10.37 & 5.068 & 5.472 & \textbf{6.377} & 8.626 & 6.214 & 8.416 \\
CR-w40 & 0.895 & 0.940 & 0.979 & 0.994 & \textbf{0.966} & 0.989 & 0.919 & 0.928 & 0.963 & 0.987 & \textbf{0.950} & 0.980 \\
AL-w40 & 5.202 & 6.809 & 8.263 & 10.95 & \textbf{7.933} & 10.63 & 5.066 & 5.598 & 6.618 & 8.849 & \textbf{6.412} & 8.621 \\
CR-w50 & 0.890 & 0.942 & 0.981 & 0.995 & \textbf{0.968} & 0.990 & 0.914 & 0.927 & 0.965 & 0.988 & \textbf{0.952} & 0.981 \\
AL-w50 & 5.238 & 6.968 & 8.446 & 11.00 & \textbf{8.093} & 10.69 & 5.063 & 5.678 & 6.767 & 8.876 & \textbf{6.530} & 8.658 \\[5pt]

Model-3 &   &  \\[2pt]

CR-w10 & 0.923 & 0.935 & 0.956 & 0.980 & \textbf{0.954} & 0.978 & 0.919 & 0.945 & 0.957 & 0.972 & \textbf{0.959} & 0.976 \\
AL-w10 & 4.828 & 3.945 & 4.563 & 5.723 & \textbf{4.522} & 5.744 & 4.450 & 3.726 & 4.282 & 5.029 & \textbf{4.139} & 4.993 \\
CR-w20 & 0.909 & 0.933 & \textbf{0.964} & 0.988 & 0.957 & 0.984 & 0.914 & 0.940 & 0.959 & 0.977 & \textbf{0.956} & 0.978 \\
AL-w20 & 4.503 & 4.253 & \textbf{5.022} & 6.477 & 5.023 & 6.671 & 4.107 & 3.800 & 4.442 & 5.458 & \textbf{4.281} & 5.492 \\
CR-w30 & 0.905 & 0.935 & \textbf{0.969} & 0.991 & 0.961 & 0.986 & 0.911 & 0.935 & 0.959 & 0.981 & \textbf{0.954} & 0.980 \\
AL-w30 & 4.547 & 4.545 & \textbf{5.411} & 6.917 & 5.437 & 7.155 & 4.088 & 3.891 & 4.599 & 5.805 & \textbf{4.463} & 5.883 \\
CR-w40 & 0.903 & 0.937 & \textbf{0.973} & 0.993 & 0.964 & 0.987 & 0.908 & 0.930 & 0.959 & 0.983 & \textbf{0.953} & 0.980 \\
AL-w40 & 4.649 & 4.741 & \textbf{5.656} & 7.148 & 5.689 & 7.382 & 4.133 & 3.964 & 4.715 & 6.005 & \textbf{4.608} & 6.062 \\
CR-w50 & 0.902 & 0.939 & \textbf{0.975} & 0.994 & 0.966 & 0.987 & 0.907 & 0.927 & 0.959 & 0.983 & \textbf{0.952} & 0.979 \\
AL-w50 & 4.736 & 4.875 & \textbf{5.822} & 7.291 & 5.854 & 7.475 & 4.212 & 3.994 & 4.778 & 6.084 & \textbf{4.679} & 6.058 \\

\bottomrule
\end{tabular}
\vspace{2pt}

    \raggedright
    {\footnotesize Note: $PI_b$, $PI_{m}$, $PI_{sa}$, $PI_{st}$, $PI_{aa}$ and $PI_{at}$ represent different PIs as explained in the main text; CR-w10 and AL-w10 stand for the average coverage rate and length with DNN possessing a structure [10,10], respectively, and so on. We use the bold font to indicate the optimal cPI among various variants under different training sample sizes, structures of DNN, and simulation models. The optimal one is determined by selecting the variant with the shortest AL and at least $95\%$ CR. If no cPI can reach the $95\%$ CR, we select the variant with the largest CR as the optimal one.}
    \label{Tab:SimresM1-3}
\end{table}

\begin{table}[htbp]
    \centering
\caption{Average coverage rate (CR) and length (AL) of different PIs for Models 4-6.}
\begin{tabular}{lllllll|llllll}

\toprule
 & $PI_b$ & $PI_{m}$ & $PI_{sa}$ & $PI_{st}$ & $PI_{aa}$ & $PI_{at}$ & $PI_b$ & $PI_{m}$ & $PI_{sa}$ & $PI_{st}$ & $PI_{aa}$ & $PI_{at}$ \\
\midrule
 & \multicolumn{6}{c}{$n = 2000$}   & \multicolumn{6}{c}{$n = 10000$}  \\
Model-4 &    & \\[2pt]
CR-w10 & 0.944 & 0.941 & 0.976 & 0.991 & \textbf{0.961} & 0.981 & 0.943 & 0.947 & 0.973 & 0.987 & \textbf{0.962} & 0.978 \\
AL-w10 & 9.052 & 7.363 & 8.950 & 11.58 & \textbf{8.578} & 11.39 & 8.449 & 6.428 & 7.419 & 9.193 & \textbf{7.186} & 9.267 \\
CR-w20 & 0.933 & 0.946 & 0.983 & 0.994 & \textbf{0.967} & 0.983 & 0.935 & 0.939 & 0.973 & 0.990 & \textbf{0.958} & 0.979 \\
AL-w20 & 9.120 & 8.706 & 10.78 & 13.73 & \textbf{10.22} & 12.84 & 8.211 & 7.004 & 8.301 & 10.92 & \textbf{8.028} & 10.87 \\
CR-w30 & 0.927 & \textbf{0.953} & 0.987 & 0.995 & 0.971 & 0.985 & 0.929 & 0.937 & 0.976 & 0.992 & \textbf{0.959} & 0.980 \\
AL-w30 & 9.514 & \textbf{9.724} & 12.06 & 14.90 & 11.22 & 13.70 & 8.383 & 7.539 & 9.149 & 12.10 & \textbf{8.853} & 11.77 \\
CR-w40 & 0.924 & \textbf{0.956} & 0.988 & 0.995 & 0.972 & 0.985 & 0.924 & 0.935 & 0.977 & 0.992 & \textbf{0.960} & 0.978 \\
AL-w40 & 9.912 & \textbf{10.14} & 12.61 & 15.34 & 11.57 & 14.02 & 8.539 & 7.884 & 9.710 & 12.53 & \textbf{9.332} & 11.89 \\
CR-w50 & 0.921 & \textbf{0.954} & 0.987 & 0.994 & 0.970 & 0.983 & 0.921 & 0.933 & 0.977 & 0.991 & \textbf{0.958} & 0.975 \\
AL-w50 & 10.11 & \textbf{10.07} & 12.60 & 15.22 & 11.45 & 13.82 & 8.670 & 7.941 & 9.884 & 12.45 & \textbf{9.385} & 11.50 \\[5pt]
Model-5 &    &  \\[2pt]
CR-w10 & 0.938 & 0.939 & 0.972 & 0.989 & \textbf{0.959} & 0.981 & 0.941 & 0.942 & 0.966 & 0.983 & \textbf{0.957} & 0.976 \\
AL-w10 & 10.58 & 9.859 & 12.33 & 16.20 & \textbf{11.58} & 15.50 & 9.937 & 8.462 & 9.993 & 12.86 & \textbf{9.570} & 12.58 \\
CR-w20 & 0.926 & 0.947 & 0.982 & 0.992 & \textbf{0.968} & 0.984 & 0.931 & 0.938 & 0.969 & 0.988 & \textbf{0.956} & 0.980 \\
AL-w20 & 10.87 & 11.58 & 14.55 & 18.49 & \textbf{13.57} & 17.15 & 9.902 & 9.339 & 11.35 & 15.30 & \textbf{10.80} & 14.76 \\
CR-w30 & 0.922 & \textbf{0.955} & 0.986 & 0.993 & 0.973 & 0.986 & 0.926 & 0.938 & 0.974 & 0.990 & \textbf{0.960} & 0.981 \\
AL-w30 & 11.54 & \textbf{12.74} & 15.85 & 19.69 & 14.68 & 18.21 & 10.14 & 9.958 & 12.32 & 16.28 & \textbf{11.68} & 15.51 \\
CR-w40 & 0.919 & \textbf{0.959} & 0.986 & 0.993 & 0.974 & 0.987 & 0.923 & 0.938 & 0.976 & 0.990 & \textbf{0.961} & 0.980 \\
AL-w40 & 11.93 & \textbf{13.14} & 16.31 & 20.11 & 15.02 & 18.65 & 10.33 & 10.26 & 12.78 & 16.36 & \textbf{12.08} & 15.39 \\
CR-w50 & 0.916 & \textbf{0.958} & 0.986 & 0.993 & 0.973 & 0.986 & 0.919 & 0.936 & 0.976 & 0.989 & \textbf{0.961} & 0.978 \\
AL-w50 & 12.14 & \textbf{12.96} & 16.14 & 19.75 & 14.77 & 18.31 & 10.49 & 10.37 & 12.98 & 16.13 & \textbf{12.15} & 14.89 \\[5pt]
Model-6 &   &  \\[2pt]
CR-w10 & 0.966 & 0.940 & \textbf{0.971} & 0.990 & 0.958 & 0.974 & 0.963 & \textbf{0.951} & 0.970 & 0.983 & 0.967 & 0.979 \\
AL-w10 & 11.86 & 7.064 & \textbf{8.745} & 11.58 & 8.781 & 12.20 & 10.68 & \textbf{6.042} & 7.321 & 9.058 & 7.055 & 9.701 \\
CR-w20 & 0.955 & 0.943 & 0.978 & 0.993 & \textbf{0.960} & 0.975 & 0.950 & 0.941 & 0.967 & 0.986 & \textbf{0.960} & 0.978 \\
AL-w20 & 11.42 & 8.494 & 10.62 & 13.88 & \textbf{10.59} & 13.39 & 9.795 & 6.701 & 8.220 & 10.82 & \textbf{8.042} & 11.51 \\
CR-w30 & 0.952 & \textbf{0.951} & 0.985 & 0.994 & 0.964 & 0.978 & 0.945 & 0.938 & 0.970 & 0.990 & \textbf{0.960} & 0.977 \\
AL-w30 & 11.96 & \textbf{9.612} & 12.01 & 14.90 & 11.60 & 14.09 & 9.876 & 7.240 & 9.033 & 12.25 & \textbf{8.995} & 12.46 \\
CR-w40 & 0.951 & \textbf{0.954} & 0.986 & 0.994 & 0.965 & 0.978 & 0.945 & 0.936 & 0.973 & 0.991 & \textbf{0.959} & 0.974 \\
AL-w40 & 12.54 & \textbf{10.12} & 12.64 & 15.28 & 11.99 & 14.48 & 10.22 & 7.629 & 9.659 & 13.01 & \textbf{9.587} & 12.65 \\
CR-w50 & 0.950 & \textbf{0.953} & 0.986 & 0.994 & 0.965 & 0.977 & 0.945 & 0.934 & 0.975 & 0.991 & \textbf{0.956} & 0.970 \\
AL-w50 & 12.93 & \textbf{10.15} & 12.76 & 15.33 & 12.01 & 14.60 & 10.53 & 7.836 & 10.07 & 13.24 & \textbf{9.773} & 12.29 \\

\bottomrule
\end{tabular}
\vspace{2pt}

    \raggedright
    {\footnotesize Note: $PI_b$, $PI_{m}$, $PI_{sa}$, $PI_{st}$, $PI_{aa}$ and $PI_{at}$ represent different PIs as explained in the main text; CR-w10 and AL-w10 stand for the average coverage rate and length with DNN possessing a structure [10,10], respectively, and so on. We use the bold font to indicate the optimal cPI among various variants under different training sample sizes, structures of DNN, and simulation models. The optimal one is determined by selecting the variant with the shortest AL and at least $95\%$ CR. If no cPI can reach the $95\%$ CR, we select the variant with the largest CR as the optimal one.}
    \label{Tab:SimresM4-6}
\end{table}

\FloatBarrier

\section{Real Data Analysis}\label{Sec:RealData}
In this section, we propose a comprehensive empirical study with the wine quality dataset from \cite{cortez2009modeling}. This dataset contains two sub-datasets (size 1599 and 4898, respectively) about red and white wine, which are available at the UCI machine learning repository (\url{https://archive.ics.uci.edu/dataset/186/wine+quality}). For each sub-dataset, there are eleven quantitative variables as predictors and a quantitative variable, which is a score between 0 and 10 to measure the wine quality. We try to build various PIs to measure the prediction accuracy of wine quality once the values of eleven predictors are given. 

According to the simulation studies, the asymmetric calibration approach with $\widetilde{F}_{Avg}$ works best. Thus, we only consider $PI_{aa}$ in the real data analysis to simplify the presentation. In addition, as discussed in \cref{SubsubSec:finietcover}, we could apply the $PI_{aaa}$, which is an adjusted version of $PI_{aa}$ to further correct the undercoverage issue when the sample is finite. 

For kernel-based cPIs, we take the univariate Gaussian kernel to perform the estimation. Besides, we have several bandwidth-choice strategies. In this empirical study, four approaches can be acquired by considering the combinations of treating $Y$ as a continuous or ordered (discrete) variable and applying the maximum likelihood or least square cross-validation method. The performance of using least-square cross-validation to determine $\hat{f}(y|x)$ is pretty unsatisfactory when the sample size is small, so we do not present cPI results with such kernel estimators here; see \cref{B2:cross-validationmethod} for the introduction of the density estimation with two different cross-validation methods. In short, the cPI, which is based on least-square cross-validation, fails since the conditional density estimation $\hat{f}(y|x)$ can be pretty numerically unstable. So, we take the maximum likelihood cross-validation methods to build kernel estimators. For the calibration strategy, we still focus on the asymmetric calibration way. 

Additionally, we consider several PIs based on deep-generative approaches: (1) Quantile PI and Pertinent PI proposed by \cite{wu2024deep}; (2) PI based on adversarial training w.r.t. KL-divergence proposed by \cite{zhou2023deep}; (3) PI based on adversarial training w.r.t. Wasserstein-1 distance proposed by \cite{liu2021wasserstein}; see corresponding references for more details. We take the same DNN structure used in these existing works, i.e., these deep generative PIs are built with a DNN possessing one hidden layer with 50 neurons.

In total, we consider 8 cPIs based on DNN or kernel estimator, i.e., $PI_{aa}$, $PI_{aaa}$, $PI^s_{aa}$, $PI^s_{aaa}$, $PI^{cm}_{aak}$, $PI^{cm}_{aaak}$, $PI^{om}_{aak}$, $PI^{om}_{aaak}$; here $PI_{aa}$ and $PI_{aaa}$ represent the asymmetric and adjusted asymmetric calibration prediction intervals with DNN estimators, respectively; the structure of DNN is $[10,10]$; for a fair comparison, $PI^s_{aa}$ and $PI^s_{aaa}$ represent the variants of $PI_{aa}$ and $PI_{aaa}$ with all DNN structures being $[50]$ which is the same structure used in deep generative methods described above, respectively; $PI^{*}_{aak}$ is the calibration prediction interval based on kernel estimators; $*$ can be replaced by $cm$, $om$. The first letter of the up script stands for the type of the dependent random variable $Y$, i.e., $c$ and $o$, which means that we treat $Y$ as continuous and ordered-discrete variables, respectively. The second letter of the up script stands for the bandwidth-choice strategy, i.e., $m$ represents the maximum likelihood of cross-validation. To see the effects of the equal distance $D$ on the coverage rate, we consider $g = 200, 100, 50, 25, 12, 5$. 

To inspect the performance of different PIs in practice, we intend to make a small training set. Meanwhile, we hope to make a comparison between cPI and deep generative PI, so we take the same training and test sets used in \cite{wu2024deep}, i.e., 199 training data and 1120 test data for the red wine case, which can be reproduced by applying the function \textit{train$\_$test$\_$split} from \textit{sklearn.model$\_$selection} to do such splitting with \textit{random$\_$state=1}. For real datasets, we can not perform replications to compute the coverage rate and average prediction interval length in a precise way. Thus, we define the coverage rate and average length for real study as $CR_t:= \frac{1}{T_t}\sum_{i=1}^{T_t}\mathbbm{1}( Y^t_{i}\in \widehat{I}_C|X^t_{i})$ and $AL_t := \frac{1}{T_t}\sum_{i=1}^{T_t} (R_{i} - L_{i})$; $T_t$ is the size of test data; $\{Y^t_{i}, X^t_{i}\}_{i=1}^{T_t}$ is the test dataset; $\widehat{I}_C$ represents all PIs in our considerations; $R_i$ and $L_i$ are the right and left endpoints of corresponding PIs.

The results of different cPIs are summarized in \cref{Table:EmpircalredcPI}. We can find $PI_{aa}$ with $g = 12$ is the optimal one among all combinations of types and number of grid points with the potential to achieve at least a 95$\%$ coverage rate and give the shortest $AL_t$ in the meantime. For the cPI with kernel estimators, treating the dependent variable as ordered-discrete is better than treating it as a continuous variable. For the adjusted cPI, no matter whether it is based on the kernel or DNN estimators, its $CR_t$ is always larger than the corresponding variant without adjustment. Also, the $CR_t$ increases eventually as the number of grid points decreases.

Furthermore, we compare each optimal cPI with DNN estimators and deep-generative PIs in \cref{Table:cPIgenPIred}. We should mention that the process to build the $QPI$, $PPI$, $PI_{KL}$ and $PI_{WA}$ involves tuning the dimension of a so-called reference random variable. We select optimal results of these deep-generative PIs from \cite{wu2024deep} with the criterion to determine the optimal PI whose $CR_t$ is as close to 0.95 as possible and the length is the shortest meanwhile length. Also, we should notice that the coverage rate of $PI_{KL}$ and $PI_{WA}$ varies severely when the tuning parameters (the dimension of the reference variable) change. However, the cPI is more robust against the tuning parameter (the number of grid points). Overall, $PI_{aa}$ beats deep-generative PIs for red wine data.

\begin{table}[htbp]
    \centering
    \small
    \caption{The comparison of different asymmetric calibration PIs on predicting red wine data.}
\begin{tabular}{lllllllllllll}
\toprule
& \multicolumn{2}{c}{$g = 200$}   & \multicolumn{2}{c}{$g = 100$} & \multicolumn{2}{c}{$g = 50$} & \multicolumn{2}{c}{$g = 25$} & \multicolumn{2}{c}{$g = 12$} & \multicolumn{2}{c}{$g = 5$} \\[2pt]
& $CR_t$ & $AL_t$  & $CR_t$ & $AL_t$  & $CR_t$ & $AL_t$  & $CR_t$ & $AL_t$  & $CR_t$ & $AL_t$ & $CR_t$ & $AL_t$   \\
\midrule
Red-wine & \\[2pt]
$PI_{aa}$   &  0.968 &  3.223 &  0.960 &  2.891 &  0.955 &  2.771 &  0.968 &  2.954 &  \textbf{0.952} &  \textbf{2.618} & 0.930 & 3.654 \\[1pt]
$PI_{aaa}$  &  0.971 &  3.270 &  0.971 &  2.985 &  \textbf{0.971} &  \textbf{2.957} &  0.971 &  3.313 &  0.964 &  3.424 & 0.980 & 4.781\\[1pt]
$PI^s_{aa}$  & 0.973 &4.249 &  0.976 &  4.047 &  0.961 &  3.504 &  0.962 &  3.637 &  \textbf{0.962} &  \textbf{3.109} & 0.952 & 3.930 \\[1pt]
$PI^s_{aaa}$  & 0.976 & 4.272 &  0.978 &  4.114 &  \textbf{0.971} &  \textbf{3.633} &  0.967 &  3.887 &  0.973 &  3.823 & 0.980 & 4.872   \\[1pt]
$PI^{cm}_{aak}$  &  0.928 &  1.848 &  0.929 &  1.873 &  0.930 &  1.912 &  0.936 &  2.018 &  \textbf{0.937} &  \textbf{2.035} & 0.937 & 3.156\\[1pt]
$PI^{cm}_{aaak}$ &  0.932 &  1.896 &  0.936 &  1.968 &  0.936 &  2.104 &  0.937 &  2.408 &  0.938 &  2.886 & \textbf{0.981} & \textbf{4.896} \\[1pt]
$PI^{om}_{aak}$ &   \textbf{0.950} &  \textbf{3.639} &   0.950 &  3.652 &  0.952 &  3.674 &  0.954 &  3.733 &  0.975 &  3.872 & 0.972 & 4.300 \\[1pt]
$PI^{om}_{aaak}$ &  \textbf{0.954} &  \textbf{3.664} &  0.954 &  3.702 &  0.957 &  3.775 &  0.971 &  3.938 &  0.981 &  4.317 & 0.983 & 4.968 \\[1pt]
\bottomrule
\end{tabular}
    \vspace{2pt}
    
    \raggedright
    { \footnotesize Note: We use the bold font to indicate the optimal performance for each type of cPI among different numbers of grid points. The optimal one is determined by selecting the variant with the shortest $AL_t$ and at least $95\%$ $CR_t$. If one type of cPI can not reach the $95\%$ $CR_t$, we select the variant with the largest $CR_t$ as the optimal one for this type of cPI.}
    \label{Table:EmpircalredcPI}
\end{table}

\begin{table}[htbp]
    \centering
    \caption{The comparison between deep-generative based PIs and each optimal calibration PIs with DNN estimators on red wine data.}
    \begin{tabular}{lccccccccccc}
    \toprule
    &  $QPI$ & $PPI$ & $PI_{KL}$ & $PI_{WA}$ & $PI_{aa}$  & $PI_{aaa}$  &  $PI^s_{aa}$  & $PI^s_{aaa}$  \\
    \midrule
      $CR_t$   &  0.937 & 0.942 & 0.950 & 0.952 & \textbf{0.952} & 0.971 & 0.962 & 0.971 \\
      $AL_t$   &  2.631 & 2.859 & 3.285 & 3.369 & \textbf{2.618} & 2.957 & 3.109 & 3.633\\
      \bottomrule
    \end{tabular}
      \vspace{2pt}
    
    \raggedright
    { \footnotesize Note: The optimal PI is indicated by the bold font. The choice criterion is the same as the method used in \cref{Table:EmpircalredcPI}.}
    \label{Table:cPIgenPIred}
\end{table}
\FloatBarrier

For white wine data, we still take the same training and test sets used in existing work, i.e., training size 195 and test size 3763, which can be reproduced by applying the function \textit{train$\_$test$\_$split} from \textit{sklearn.model$\_$selection} to do such splitting with \textit{random$\_$state=1}. The comparison of different cPIs is presented in \cref{Table:EmpircalwhitecPI}. In this case, the naive asymmetric cPI struggles to return a satisfied $CR_t$. On the other hand, adjusted asymmetric cPIs can guarantee at least a 95$\%$ coverage rate as the number of grid points varies. Considering the overall performance regarding $CR_t$ and $AL_t$, $PI_{aaa}^s$ and $PI_{aak}^{om}$ are indistinguishable and stand for superior performance compared to other cPIs. The comparison of cPIs with DNN estimators and deep generative PIs is tabulated in \cref{Table:cPIgenPIwhite}. $PI_{aaa}^s$ is the optimal one since its length is the shortest among all PIs that can guarantee the coverage rate. The $PPI$ can be an alternative candidate in practice since its length is less than $PI_{aaa}^s$, though its coverage rate is slightly lower than the nominal level.

\begin{table}[htbp]
    \centering
    \small
    \caption{The comparison of different asymmetric calibration PIs on predicting white wine data.}
\begin{tabular}{lllllllllllll}
\toprule
& \multicolumn{2}{c}{$g = 200$}   & \multicolumn{2}{c}{$g = 100$} & \multicolumn{2}{c}{$g = 50$} & \multicolumn{2}{c}{$g = 25$} & \multicolumn{2}{c}{$g = 12$} & \multicolumn{2}{c}{$g = 5$} \\[2pt]
& $CR_t$ & $AL_t$  & $CR_t$ & $AL_t$  & $CR_t$ & $AL_t$  & $CR_t$ & $AL_t$  & $CR_t$ & $AL_t$  & $CR_t$ & $AL_t$\\
\midrule
White-wine & \\[2pt]
$PI_{aa}$  &  0.917 &   3.23 &  0.921 &  2.885 &  0.907 &  3.056 &  0.917 &  2.652 &  \textbf{0.927} &  \textbf{2.907} &  0.893 &  3.304 \\[1pt]
$PI_{aaa}$  &  0.931 &  3.269 &  0.926 &  2.963 &  0.929 &  3.223 &  0.927 &  2.969 &  0.942 &  3.622 &  \textbf{0.958} &  \textbf{4.877} \\ [1pt]
$PI^s_{aa}$  &  \textbf{0.943} &  \textbf{4.079} &  0.934 &  3.426 &  0.921 &  3.367 &  0.927 &    3.400 &  0.912 &  3.573 & 0.909 & 3.881 \\[1pt]
$PI^s_{aaa}$  &  \textbf{0.952} &  \textbf{4.124} &  0.943 &  3.517 &  0.934 &   3.530 &  0.936 &  3.668 &  0.929 &  4.025 & 0.957 & 4.854\\[1pt]
$PI^{cm}_{aak}$  &   0.910 &   2.230 &  0.911 &  2.252 &  0.912 &  2.299 &  0.915 &  2.358 &  \textbf{0.919} &  \textbf{2.427} &  0.919 &  3.505 \\[1pt]
$PI^{cm}_{aaak}$ &  0.912 &  2.275 &  0.915 &  2.342 &  0.919 &  2.478 &  0.919 &  2.721 &   0.920 &   3.220 &  \textbf{0.956} &  \textbf{4.904} \\[1pt]
$PI^{om}_{aak}$  & 0.937 & 3.661    & 0.937 & 3.675 & 0.937 & 3.700  & 0.939 & 3.755 & \textbf{0.950} & \textbf{3.882} & 0.946 & 4.301 \\[1pt]
$PI^{om}_{aaak}$  & 0.937 & 3.686    & 0.939 & 3.724 & 0.941 & 3.801 & 0.946 & 3.960  & \textbf{0.956} & \textbf{4.328} & 0.958 & 4.988 \\
\bottomrule
\end{tabular}
    \vspace{2pt}
    
    \raggedright
    { \footnotesize Note: We use the bold font to indicate the optimal performance for each type of cPI among different numbers of grid points. The optimal one is determined by selecting the variant with the shortest $AL_t$ and at least $95\%$ $CR_t$. If one type of cPI can not reach the $95\%$ $CR_t$, we select the variant with the largest $CR_t$ as the optimal one for this type of cPI.}

    \label{Table:EmpircalwhitecPI}
\end{table}

\begin{table}[htbp]
    \centering
    \caption{The comparison between deep-generative based PIs and each optimal calibration PIs with DNN estimators on white wine data.}
    \begin{tabular}{lccccccccccc}
    \toprule
    &  $QPI$ & $PPI$ & $PI_{KL}$ & $PI_{WA}$ & $PI_{aa}$  & $PI_{aaa}$  &  $PI^s_{aa}$  & $PI^s_{aaa}$  \\
    \midrule
      $CR_t$   &  0.824 & 0.944 & 0.985 & 0.958 & 0.927 & 0.958 & 0.943 & \textbf{0.952} \\
      $AL_t$   &  2.890 & 3.466 & 4.745 & 4.869 & 2.907 & 4.877 & 4.079 & \textbf{4.124}\\
      \bottomrule
    \end{tabular}
          \vspace{2pt}
    
    \raggedright
    { \footnotesize Note: The optimal PI is indicated by the bold font. The choice criterion is the same as the method used in \cref{Table:EmpircalwhitecPI}.}
    \label{Table:cPIgenPIwhite}
\end{table}

Combining the two pieces of real data results, we summarize several conclusions below:
\begin{itemize}
    \item Various cPIs work well with real data, even for a small training size. Notably, adjusted cPIs always ensure the coverage rate if an appropriate number of grid points is applied. On the other hand, it is hard for the naive cPIs to achieve the nominal coverage level for the white wine data. The cPIs with DNN estimators generally perform better than the cPIs with kernel estimators.
    \item For deep generative PIs and cPIs with DNN estimators, where DNN estimators are involved in both, cPI is better overall. Although the determination of cPI requires $g$ number of DNNs, these training processes can be accomplished in parallelly. We remark that the training procedure of PPI is still computationally heavy since multiple times of retrainings of DNN with simulated data are required. 
    \item When the training sample size is small, we recommend $PI_{aaa}$, which shows great potential to achieve the desired coverage rate while avoiding excessively large interval lengths. We should notice that although a tuning procedure is still needed for $PI_{aaa}$, other methods also require an unavoidable tuning step that could be even more complicated and crucial for the final performance. 
\end{itemize}

\FloatBarrier
\section{Conclusion}\label{Sec:Conclusion}
In this paper, a systematic way to build a so-called calibration Prediction interval (cPI) is discussed. The validity of cPI does not require linearity and normality assumptions. The cPI also compensates for the estimation variability which is not usually directly captured. The whole procedure to determine a cPI consists of three main steps: (1) Apply some estimators (DNN or kernel function) to estimate the conditional CDF of $Y$ given $X_f$ at user-chosen grid points; (2) If it is necessary, perform the correction for monotonicity; (3) Take a calibration strategy to construct cPI. We give the theoretical validation of cPIs based on DNN and kernel estimators. More specifically, we carefully discuss three levels of coverage guarantee: (a) The asymptotic coverage guarantee which is satisfied by cPI with DNN and kernel estimators; (b) The large-sample coverage rate guarantee with a high probability which is satisfied by cPI with kernel and we conjecture it also holds for cPI with DNN; (c) The finite-sample coverage guarantee which is possible for cPI with DNN under several conditions. The practicality of cPI is checked by a comprehensive simulation study. Also, two empirical datasets (red and white wine datasets) are taken to verify the performance of cPI with a short sample. In conclusion, the adjusted cPI gives the advantage of overcoming the undercoverage issue for scenarios with short samples. The cPI with DNN generally outperforms the variant with kernel estimators. Moreover, cPI with DNN is competitive with other deep generative PIs. 

\acks{This numerical work was done using computational services provided by the OSG Consortium \cite{osg07,osg09,osg2,osg1}, which is supported by the NSF awards 2030508 and 1836650. This research was partially supported by NSF grant DMS 24-13718.
}
\FloatBarrier

\appendix

\section{\textsc{Proof}}\label{Appendix:A}
\begin{proof}[\textsc{\textbf{Proof of \cref{Lemma:CDE}}}]
We first give the proof of the first part of \cref{Lemma:CDE}. Let $H_0(x) = \mathbb{E}(G(x,\epsilon))$. We show that $H_0(X)$ is indeed the conditional expectation of $Y$ given $X$. By the definition of conditional expectation, we need to show:
\begin{itemize}
    \item[(1)] $H_0(X)$ is $\sigma(X)$-measurable; 
    \item[(2)] $\int_A H_0(X)d\mathbb{P} = \int_A G(X,\epsilon)d\mathbb{P}$ for $\forall A \in \sigma(X)$.
\end{itemize}
For (1), it is trivial since $X$ is $\sigma(X)$-measurable and $H_0(X)$ is a measurable function of $X$. Therefore, $H_0(X)$ is also $\sigma(X)$-measurable. 

For (2), we first simplify the integral $\int_A H_0(X)d\mathbb{P}$ by introducing an indicator random variable $\Xi:=\mathbbm{1}(A)$ which means $\Xi = 1$ if the event from sample space belongs to $A$ and $\Xi = 0$ otherwise. We define the joint probability distribution of $(X,\Xi)$ as $\nu$ and the probability distribution of $\epsilon$ as $\varrho$. Then, we can get 
$$
\int_A H_0(X) d\mathbb{P} = \int H_0(X)\Xi d\mathbb{P} = \int H_0(x)\xi d\nu(x,\xi).
$$
For the last equality, we switch the integration variable space from $\Omega$ to $\mathbb{R}^{d+1}$. We can further get 
$$ 
\int H_0(x)\xi d\nu(x,\xi) = \int \mathbb{E}(G(x,\epsilon)) \xi d\nu(x,\xi) = \int \int G(x,\epsilon) \xi d\varrho(\epsilon) d\nu(x,\xi).
$$
Since we assume $\epsilon$ and $X$ are independent under A2, the integration variable space for the above equalities is $\mathbb{R}^{d+2}$. Also, $G$ is integrable under A1, so we can get the above equalities by Fubini's theorem. Finally, we can conclude that:
$$
\int \int G(x,\epsilon) \xi d\varrho(\epsilon) d\nu(x,\xi) = \mathbb{E}[G(X, \epsilon)\Xi] = \int G(X, \epsilon) \Xi d \mathbb{P}.
$$
For the last equality, we transform the integration variable space back to $\Omega$ again. Consequently, we can get 
$$
\int_A H_0(X) d\mathbb{P} = \int G(X, \epsilon) \Xi d \mathbb{P} = \int_A G(X, \epsilon) d\mathbb{P} = \int_A Y d\mathbb{P};
$$
which shows $H_0(X) = \mathbb{E}(G(X,\epsilon))$ is indeed the conditional expectation of $Y$ given $X$. 

Secondly, we show the function $H_0(x) = \mathbb{E}(G(x,\epsilon))$ is continuous w.r.t. the argument $x$. From the derivation above, we know 
$$
\mathbb{E}(G(x,\epsilon)) = \int G(x,z)f_\epsilon(z) dz. 
$$

Under A1, $G(x,\epsilon)$ is a uniformly continuous function, i.e., for $\forall \varepsilon>0$, $\exists \delta>0$ such that 
$$
|G(u, \epsilon)-G(v, \epsilon)|<\varepsilon \text { for all }(u, \epsilon)\text{ , }(v, \epsilon) \in I \text { with }|u-v|<\delta;
$$
where $I$ is the joint space of $X$ and $\epsilon$, i.e., $I\subseteq \mathbb{R}^{d+1}$. Therefore, for these $u, v \in \mathbb{R}^d$ with $|u-v|<\delta$ we can get

$$
|H_0(u)-H_0(v)| \leq \int|G(u, \epsilon)-G(v, \epsilon)|f_\epsilon d x \leq  \varepsilon \cdot  M \cdot  \infty;
$$
where $M \geq |f_\epsilon|$ since $f_\epsilon$ is bounded under A2. Since $\varepsilon$ can be arbitrarily small, we get that $|H_0(u)-H_0(v)| \leq \varepsilon^{\prime}$ for $\forall \varepsilon^{\prime} > 0$. Thus, we can conclude that $H_0(x)$ is a continuous function w.r.t. $x$. 
\end{proof}

\begin{proof}[\textsc{\textbf{Proof of \cref{Lemma:EstHk}}}]
Similarly to the proof of Lemma 2.1, we show $H_j(X) = \mathbb{E}(\mathbbm{1}(Y\leq q_j)|X) = \mathbb{E}(\mathbbm{1}(G(X,\epsilon)\leq q_j))$ is the conditional expectation of $Z:= \mathbbm{1}(Y\leq q_j)$ given $X$ for any $j$. First, $H_j(X)$ is still $\sigma(X)$-measurable since the $\mathbbm{1}(G(X,\epsilon)$ is measurable function of $X$. Under A1 and A2, if we denote $\mathbbm{1}(G(X,\epsilon)\leq q_j)$ by $\widetilde{G}(X,\epsilon,q_j)$ where $q_j$ is a constant. Again, by defining an indicator random variable $\Xi:=\mathbbm{1}(A)$, we can show:
$$
\int_A H_j(X) d\mathbb{P} = \int H_j(X)\Xi d\mathbb{P} = \int \int \widetilde{G}(x,\epsilon,q_j) \xi d\varrho(\epsilon) d\nu(x,\xi) =  \int_A \widetilde{G}(X, \epsilon,q_j) d\mathbb{P} = \int_A Z d\mathbb{P}.
$$
Besides, it is trivial that $H_j(x)$ is a continuous function w.r.t. $x$ under A3.

\end{proof}

\begin{proof}[\textsc{\textbf{Proof of \cref{Theorem:coverageasy}}}]
The proof outline relies on the MSE error bound of the DNN estimator on uniformly continuous functions. Since we care about the asymptotic coverage in this theorem, the convergence rate of the error bound does not matter for the proof. Under A1 to A3, by \cref{Lemma:CDE} and \cref{Lemma:EstHk}, we know that the conditional mean and conditional CDF are all continuous w.r.t. $x$. Moreover, under B1, it is trivial that these oracle functions are uniformly continuous w.r.t. $x$. By \cref{Theorem:errorboundofC0f}, we have
$$
\mathbb{E}\left( (\widehat{H}_j(X) - H_j(X))^2 \right)\to 0,~\text{with a probability converges to~} 1, ~\text{as}~n\to\infty;
$$
for $j = 0,\ldots,g$; $H_0$ represent the oracle conditional mean function. In other words, we have
\begin{equation}\label{Eq:ConditionalExpecConver}
\int (\widehat{H}_j(x) - H_j(x))^2f(x)dx \to 0,~\text{with a probability converges to~} 1, ~\text{as}~n\to\infty;
\end{equation}
If we further assume B2 ,namely the marginal density $f(x)$ is positive on its domain $\mathcal{X}$, then \cref{Eq:ConditionalExpecConver} implies that $\widehat{H}_j(x_f) \overset{}{\to} H_j(x_f)$ with a probability tends to 1 as $n\to\infty$ for all $x_f \in \mathcal{X}$ and all $j\in\{1,\ldots,g\}$. Moreover, the randomness of $\widehat{H}_j(x_f)$ depends on our samples $S_n: = \{(X_i, Y_i)\}_{i=1}^n$. Equivalently, for the sample path $\mathbb{P}(S_n\notin A_n ) = o(1)$, $\widehat{H}_j(x_f) \to H_j(x_f)$ as $n\to\infty$, where $A_n\subseteq \mathcal{X}\times\mathcal{Y}$ is an expanding set as $n$ increasing. 

For all cPIs, their endpoints will be determined by corresponding indices s.t., the coverage rate is guaranteed to be at least $1-\alpha$. Under the consistency, $\widehat{H}_j$ can be arbitrarily close to the oracle function $H_j$ with a high probability indicated by the above results. Thus, the indices selection procedure in all algorithms will work. 

\end{proof}

\begin{proof}[\textsc{\textbf{Proof of \cref{Theorem:finiteCoverage}}}]
    The proof of \cref{Theorem:finiteCoverage} relies on the approximation result of DNN on a uniformly continuous function specified in \cref{Lemma:DNNC0f}. To verify the finite-sample coverage of $PI_{aaa}$, we focus on showing $\widetilde{F}_{\text{Avg}}(q_{\Tilde{l}}) \leq \alpha/2$ is always true under our assumptions and the condition made in \cref{Remakr:2ndcondition}; where the index $\Tilde{l}$ is determined by \cref{Algo:CaliPIaaadjusted}. The proof of showing $\widetilde{F}_{\text{Avg}}(q_{\Tilde{r}}) \geq 1- \alpha/2$ being true is similar. 
    
    This claim can be shown depending on the diagram illustrated in \cref{fig:diagramTheo2} below, where $q_{i-1}$, $q_i$ and $q_{i+1}$ are three adjacent grid points. We hope $\widehat{H}_{i-1}$, $\widehat{H}_{i}$ and $\widehat{H}_{i+1}$ are as close to $H_{i-1}$, $H_{i}$ and $H_{i+1}$ as possible. As pointed out by \cref{Theorem:coverageasy}, $\widetilde{F}_{\text{Avg}}(q_{\Tilde{l}}) < \alpha/2$ is true once $n\to\infty$. However, when the sample size is finite, there may be a relatively large error between $\widehat{H}_{i}$ and $H_i$. This may cause the chosen index is not sufficient to give a correct coverage rate for cPI. For example, if the estimated $\widehat{H}_{i+1}$ is less than $\alpha/2$, this gives us the information that the left index should be $i+1$ based on the non-adjusted calibration algorithm \cref{Algo:CaliPIaa} in finite-sample scenarios. However, it is easy to see that the ideal index is $i$, i.e., $q_i$ should be the left endpoint of a cPI. Under B5 and the adjusted algorithm \cref{Algo:CaliPIaaadjusted}, the index $\Tilde{l}$ always guarantee $\widetilde{F}_{\text{Avg}}(q_{\Tilde{l}}) \leq \alpha/2$. To see this, assuming that $\widehat{H}_{i+1}$ is unfortunately less than $\alpha/2$, by \cref{Algo:CaliPIaaadjusted}, we will choose the index $i+1-1$ to determine the left endpoints which correct the coverage rate of PI in this case. Also, under B5, there will be no situation such that $\widehat{H}_{i+1}<\alpha/2<H_i$, which means the adjusted index $\Tilde{l}$ is still not sufficient to guarantee the correct coverage.

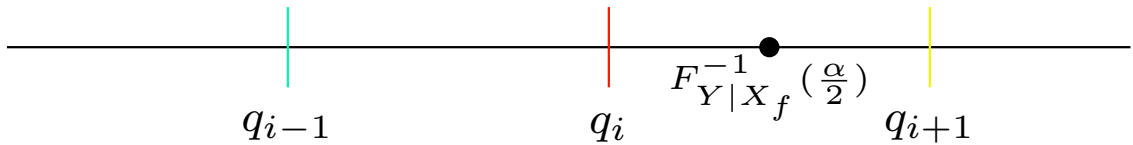
\begin{figure}[!ht]
\centering
\resizebox{1\textwidth}{!}{%
\begin{circuitikz}
\tikzstyle{every node}=[font=\scriptsize]
\draw (6.25,8.25) to[short] (13.25,8.25);
\draw [ color={rgb,255:red,5; green,245; blue,173}, ](8,8.5) to[short] (8,8);
\draw [ color={rgb,255:red,245; green,29; blue,5}, ](10,8.5) to[short] (10,8);
\draw [ color={rgb,255:red,241; green,245; blue,5}, ](12,8.5) to[short] (12,8);
\node [font=\scriptsize] at (8,7.75) {$q_{i-1}$};
\node [font=\scriptsize] at (10,7.75) {$q_{i}$};
\node [font=\scriptsize] at (12,7.75) {$q_{i+1}$};
\node at (11,8.25) [circ] {};
\node [font= \tiny] at (11,8) { $F^{-1}_{Y|X_f}(\frac{\alpha}{2})$ };
\end{circuitikz}
}
\caption{The diagram of one possible location of $\alpha/2$ quantile of the conditional distribution and three adjacent grid points for some index $i\in\{1,\ldots,g\}$.}
\label{fig:diagramTheo2}
\end{figure}

\end{proof}

\begin{proof}[\textsc{\textbf{Proof of \cref{Lemma:Errorboundofkernel}}}]
The proof is partially motivated by Proof of Theorem 1.4 in \cite{li2007nonparametric}
    To show error bound (\ref{Eq:kernelesterror}), we consider:
    \begin{equation}
        \begin{split}
            &\sup_{x,y}|\hat{f}(y|x) - f(y|x)|\\
            &= \sup_{x,y}\left| \frac{1}{\hat{f}(x)} \left( \hat{f}(x,y) - \hat{f}(x)f(y|x)\right)\right|\\
            & = \sup_{x,y}\left| \frac{1}{\hat{f}(x)} \left( \hat{f}(x,y) - \mathbb{E}(\hat{f}(x,y)) + \mathbb{E}(\hat{f}(x,y)) - \mathbb{E}(\hat{f}(x)f(y|x)) + \mathbb{E}(\hat{f}(x)f(y|x)) - \hat{f}(x)f(y|x)\right)\right|\\
            & \leq C_1 \cdot \sup_{x,y}\left|  \hat{f}(x,y) - \mathbb{E}(\hat{f}(x,y)) + \mathbb{E}(\hat{f}(x,y)) - \mathbb{E}(\hat{f}(x)f(y|x)) + \mathbb{E}(\hat{f}(x)f(y|x)) - \hat{f}(x)f(y|x)   \right| \\
            & \leq C_1 \cdot \left [\sup_{x,y}\left|  \hat{f}(x,y) - \mathbb{E}(\hat{f}(x,y)) \right| + \sup_{x,y}\left| \mathbb{E}(\hat{f}(x,y)) - \mathbb{E}(\hat{f}(x)f(y|x))     \right| +  \sup_{x,y} \left|  \mathbb{E}(\hat{f}(x)f(y|x)) - \hat{f}(x)f(y|x)  \right|   \right ] \\
            & = C_1 \cdot ( Q_1 + Q_2 + Q_3);
        \end{split}
    \end{equation}
where $Q_1 : = \sup_{x,y}\left|  \hat{f}(x,y) - \mathbb{E}(\hat{f}(x,y)) \right|$, $Q_2 := \sup_{x,y}\left| \mathbb{E}(\hat{f}(x,y)) - \mathbb{E}(\hat{f}(x)f(y|x))     \right|$; $Q_3 := \sup_{x,y} \left|  \mathbb{E}(\hat{f}(x)f(y|x)) \right.$ $\left. - \hat{f}(x)f(y|x)  \right|$. 
Then, we will analyze the order of $Q_1$, $Q_2$ and $Q_3$ one by one.

First, we consider $Q_1$. Under B1, the domain of $X$ and $Y$ is a compact set $[-M_1,M_1 ] \times [-M_2, M_2]^d$. Without loss of generality, we assume $M_1 > M_2$ and then define the space $\mathcal{S} : = [-M_1,M_1 ]^{d+1}$ which can be covered by a finite number $J_n$ cubes $\Xi_k$ with centers $\xi_{k,n}$ for $k = 1,\ldots, J_n$; $J_n$ is a constant which depends on the dimension $d+1$, the constant $M_1$, and the length of the cube $l_n$. Then, we have:
\begin{equation}\label{Eq:Q1statrt}
    \begin{split}
        \sup _{x \times y \in \mathcal{S}}|\hat{f}(x,y)-\mathbb{E}(\hat{f}(x,y))|= & \max _{1 \leq k \leq J_n} \sup _{x \times  y \in \mathcal{S} \cap \Xi_k}|\hat{f}(x,y)-\mathbb{E}(\hat{f}(x,y))| \\ 
        \leq & \max _{1 \leq k \leq J_n} \sup _{x \times y \in \mathcal{S} \cap \Xi_k}\left|\hat{f}(x,y)-\hat{f}\left(\xi_{k, n}\right)\right| \\ 
        & +\max _{1 \leq k \leq J_n}\left|\hat{f}\left(\xi_{k, n}\right)-\mathbb{E}\left(\hat{f}\left(\xi_{k, n}\right)\right)\right| \\ 
        & +\max _{1 \leq k \leq J_n} \sup _{x \times y \in \mathcal{S} \cap \Xi_k}\left|\mathbb{E}\left(\hat{f}\left(\xi_{k, n}\right)\right)-\mathbb{E}(\hat{f}(x,y))\right| \\ 
         = & V_1+V_2+V_3.
    \end{split}
\end{equation}
For the terms $V1: = \max _{1 \leq k \leq J_n} \sup _{x \times y \in \mathcal{S} \cap \Xi_k}\left|\hat{f}(x,y)-\hat{f}\left(\xi_{k, n}\right)\right|$, we can write it as 
\begin{equation}\label{Eq:V1term}
    \begin{split}
       & \max _{1 \leq k \leq J_n} \sup _{x \times y \in \mathcal{S} \cap \Xi_k}\left|\hat{f}(x,y)-\hat{f}\left(\xi_{k, n}\right)\right|\\
       & =  \max _{1 \leq k \leq J_n} \sup _{x \times y \in \mathcal{S} \cap \Xi_k}\left|   \frac{1}{nh^{d+1}} \sum_{i=1}^n \widetilde{K}\left( \Lambda^{-1} (z-Z_i)\right) -  \frac{1}{nh^{d+1}} \sum_{i=1}^n \widetilde{K}\left(\Lambda^{-1} (\xi_{k,n} - Z_i)\right)   \right| \\
       & =  \max _{1 \leq k \leq J_n} \sup _{x \times y \in \mathcal{S} \cap \Xi_k}\left| \frac{1}{nh^{d+1}} \left( \sum_{i=1}^n \nabla \widetilde{K}(\bar{z}_i)^T \cdot \Lambda^{-1} (\xi_{k,n} - Z_i)   \right)        \right|  \\
       & \leq \frac{C_2}{h^{d+2}}l_n;
    \end{split}
\end{equation}
where the first equality is due to the definition of kernel estimator; the second equality is due to the Lipschitz condition of $\widetilde{K}$ and $\bar{z}_i$ is some point on the line segment which connects $\Lambda^{-1} (z-Z_i)$ and $\Lambda^{-1} (\xi_{k,n} - Z_i)$; the last inequality is due to assumption D1 and the finite cover of the space $\mathcal{S}$; $C_2$ is some constant depends on the gradient of $\widetilde{K}$; we still apply $z$ and $Z_i$ to denote the vector $(x^T,y)^T$ and random vector $(X_i^T, Y_i)^T$, respectively. Correspondingly, $f_Z(z)$ stands for the joint density of $(X^T,Y)^T$.

For the term $V3: = \max _{1 \leq k \leq J_n} \sup _{x \times y \in \mathcal{S} \cap \Xi_k}\left|\mathbb{E}\left(\hat{f}\left(\xi_{k, n}\right)\right)-\mathbb{E}(\hat{f}(x,y))\right|$, we have:
\begin{equation}\label{Eq:V3term}
    \begin{split}
        & \max _{1 \leq k \leq J_n} \sup _{x \times y \in \mathcal{S} \cap \Xi_k}\left|\mathbb{E}\left(\hat{f}\left(\xi_{k, n}\right)\right)-\mathbb{E}(\hat{f}(x,y))\right|\\
        & \leq \max _{1 \leq k \leq J_n} \sup _{x \times y \in \mathcal{S} \cap \Xi_k}\mathbb{E} \left[ \left| \hat{f}\left(\xi_{k, n}\right) - \hat{f}(x,y) \right|  \right] \\
        & \leq \max _{1 \leq k \leq J_n}  \mathbb{E} \left[ \sup _{x \times y \in \mathcal{S} \cap \Xi_k} \left| \hat{f}\left(\xi_{k, n}\right) - \hat{f}(x,y) \right|  \right] \\
        & \leq \frac{C_3}{h^{d+2}}l_n;
    \end{split}
\end{equation}
where the last inequality is due to the inequality (\ref{Eq:V1term}); $C_3$ is another constant. The choice of $l_n$ will clear once we figure out the order of the term $V_2$. 

For $V_2 := \max _{1 \leq k \leq J_n}\left|\hat{f}\left(\xi_{k, n}\right)-\mathbb{E}\left(\hat{f}\left(\xi_{k, n}\right)\right)\right|$, we can define a function $\Psi_n(z) : = \hat{f}(z) - \mathbb{E}(\hat{f}(z)) = \sum_{i=1}^n (nh^{d+1})^{-1}[\widetilde{K}(\Lambda^{-1} (z-Z_i)) - \mathbb{E}( \widetilde{K}(\Lambda^{-1} (z-Z_i)) ) ]  = \sum_{i=1}^n \psi_{n,i}$. Then, for any $\eta > 0$, we can write 
\begin{equation}\label{Eq:V2inequality}
    \begin{split}
        \mathbb{P}( V_2 > \eta ) &= \mathbb{P} \left( \max _{1 \leq k \leq J_n}\left|\hat{f}\left(\xi_{k, n}\right)-\mathbb{E}\left(\hat{f}\left(\xi_{k, n}\right)\right)\right| > \eta \right)  \\
         & = \mathbb{P} \left( \max _{1 \leq k \leq J_n} | \Psi_n(\xi_{k,n}) | >\eta  \right)  \\
         & \leq \mathbb{P}\left( \left| \Psi_n(\xi_{1,n})\right|> \eta \bigcup  \left|\Psi_n(\xi_{2,n})\right|> \eta \bigcup  \cdots \left|\Psi_n(\xi_{J_n,n})\right|> \eta   \right) \\
         & \leq J_n \sup_{z\in\mathcal{S}} \mathbb{P}  \left(  \left| \Psi_n(z)\right|> \eta       \right).
    \end{split}
\end{equation}
Under D1, let $A_1 : = \sup_z |\widetilde{K}(z)|$. Then, we have $|\psi_{n,i}|\leq 2A_1/(nh^{d+1})$ for $i = 1,\ldots, n$. Let $\lambda_n = (nh^{d+1}\ln(n))^{1/2}$; the reason for taking such $\lambda_n$ will be explained later. Then, we have 
$$
\lambda_n|\psi_{n,i}|\leq 2A_1(\ln(n)/(nh^{d+1}))^{1/2} \leq \frac{1}{2};
$$
when $n$ is sufficiently large. Subsequently, for such a large $n$, we have
\begin{equation}\label{Eq:expinequality}
    \begin{split}
        \mathbb{E}(\exp(\pm\lambda_n\psi_{n,i})) &\leq 1 \pm \lambda_n \mathbb{E}(\psi_{n,i}) + \lambda_n^2\mathbb{E}(\psi_{n,i}^2)  \\
        & =  1  + \lambda_n^2\mathbb{E}(\psi_{n,i}^2)\\
        & \leq \exp(\mathbb{E}(\lambda_n^2\psi_{n,i}^2 ));~\text{for}~i = 1,\ldots, n;
    \end{split}
\end{equation}
where the first inequality is due to the fact that $\exp(u) \leq 1 + u + u^2$ for $|u|\leq 1/2$; the second equality is due to the fact that $\mathbb{E}(\psi_{n,i}) = 0$ for any $i$; the last inequality is due to the fact that $1+v \leq \exp(v)$ for $v\geq 0$. Therefore, we can conclude that:
\begin{equation}\label{Eq:absPsi}
    \begin{split}
        \mathbb{P}  \left(  \left| \Psi_n(z)\right|> \eta       \right) &= \mathbb{P}  \left( \left|  \sum_{i=1}^n \psi_{n,i} \right|  > \eta     \right)\\
    & =     \mathbb{P}  \left(  \sum_{i=1}^n \psi_{n,i}   > \eta  \right) +  \mathbb{P}  \left(  \sum_{i=1}^n \psi_{n,i}   < -\eta  \right) \\
    & \leq   \mathbb{P}  \left(  \sum_{i=1}^n \psi_{n,i}   > \eta  \right) +  \mathbb{P}  \left( - \sum_{i=1}^n \psi_{n,i}   >  \eta  \right) \\
    & \leq \frac{ \mathbb{E}[\exp(\lambda_n\sum_{i=1}^n \psi_{n,i} )] + \mathbb{E}[\exp(-\lambda_n\sum_{i=1}^n \psi_{n,i} )]}{\exp(\lambda_n\eta)}\\
    & \leq \frac{2\cdot \exp(\sum_{i=1}^n\mathbb{E}(\lambda_n^2\psi_{n,i}^2 ))}{\exp(\lambda_n\eta)};
    \end{split}
\end{equation}
where the fourth line is due to Markov's inequality, i.e., $\mathbb{P}(X>c)\leq \frac{\mathbb{E}(\exp(Xa))}{\exp(ac)}$ for some $a>0$; the last inequality is due to the result from inequality (\ref{Eq:expinequality}). In addition, for $\mathbb{E}(\psi^2_{n,i})$, we have
\begin{equation}\label{Eq:psinisquare}
\begin{split}
    \mathbb{E}(\psi^2_{n,i}) &= (nh^{d+1})^{-2}\mathbb{E}\left( \widetilde{K}(\Lambda^{-1} (z-Z_i)) - \mathbb{E}( \widetilde{K}(\Lambda^{-1} (z-Z_i)) ) \right)^2 \\
    & \leq (nh^{d+1})^{-2}\mathbb{E}\left( \widetilde{K}^2(\Lambda^{-1} (z-Z_i))  \right)\\
    & \leq A_2(n^2h^{d+1})^{-1};
\end{split}
\end{equation}
$A_2: = \sup_{z} \widetilde{K}^2(z)$ is some constant under D1; the first inequality is due to the fact $\mathbb{E}(U^2) - (\mathbb{E}(U))^2 = Var(U)$, taking $U$ as $\widetilde{K}(\Lambda^{-1} (z-Z_i))$; the second inequality comes from assumption D1. With result (\ref{Eq:psinisquare}), we can further simplify inequality (\ref{Eq:absPsi}) as follows:
\begin{equation}\label{Eq:absPsifinal}
     \mathbb{P}  \left(  \left| \Psi_n(z)\right|> \eta       \right) \leq 2\exp(-\lambda_n\eta + A_2\lambda_n^2(nh^{d+1})^{-1}).
\end{equation}
Since inequality (\ref{Eq:absPsifinal}) does not depend on $z$, we have:
\begin{equation}\label{Eq:absPsisup}
  \sup_{z\in\mathcal{S}} \mathbb{P}  \left(  \left| \Psi_n(z)\right|> \eta       \right) \leq 2\exp(-\lambda_n\eta + A_2\lambda_n^2(nh^{d+1})^{-1}).
\end{equation}
To determine the value of $\eta$ and $\lambda_n$, we hope $\eta$ converges to 0 as fast as possible. Meanwhile, we still need that $\lambda_n\eta$ converges to $\infty$ to make the probability $\sup_{z\in\mathcal{S}} \mathbb{P}  \left(  \left| \Psi_n(z)\right|> \eta       \right)$ converges to 0. Let $\lambda_n\eta = C_4\ln(n)$ for some constant $C_4$. Then, we need $\ln(n) \geq \lambda_n^2/(nh^{d+1})$, which implies we can take $\lambda_n = (nh^{d+1}\ln(n))^{1/2}$, so $\eta = C_4\left(\ln(n)/(nh^{d+1})\right)^{1/2}$. Thus, the inequality (\ref{Eq:absPsisup}) can be further written to
\begin{equation*}
\begin{split}
    \sup_{z\in\mathcal{S}} \mathbb{P}  \left(  \left| \Psi_n(z)\right|> \eta       \right)  &\leq 2\exp(-\lambda_n\eta + A_2\lambda_n^2(nh^{d+1})^{-1}) \\
    & = 2\exp(-C_4\ln(n) + A_2\ln(n))  \\
    & = 2\exp(-(C_4 - A_2)\ln(n)) \\
    & = 2\exp(-\alpha\ln(n))\\
    & = \frac{2}{n^\alpha};
\end{split}
\end{equation*}
where $\alpha = C_4 - A_2$. Finally, we can conclude that inequality (\ref{Eq:V2inequality}) is 
\begin{equation}\label{Eq:V2inequalityfinal}
    \begin{split}
        \mathbb{P}\left( V_2 > C_4\left(\ln(n)/(nh^{d+1})\right)^{1/2} \right) & \leq J_n \sup_{z\in\mathcal{S}} \mathbb{P}  \left(  \left| \Psi_n(z)\right|> C_4\left(\ln(n)/(nh^{d+1})\right)^{1/2}      \right) = \frac{2J_n}{n^\alpha}.
    \end{split}
\end{equation}
Furthermore, we have 
\begin{equation}\label{Eq:V2inequalityfinalleq}
        \mathbb{P}\left( V_2 \leq C_4\left(\ln(n)/(nh^{d+1})\right)^{1/2} \right) \geq 1 -  \frac{2J_n}{n^\alpha}.
\end{equation}
To incorporate all terms $V_1$, $V_2$ and $V_3$ together, we can take $l_n = (\ln(n)h^{d+3}/n)^{1/2}$. Then, from inequalities (\ref{Eq:V1term}) and (\ref{Eq:V3term}) satisfy
\begin{equation}\label{Eq:V1V3final}
\begin{split}
    V1: &= \max _{1 \leq k \leq J_n} \sup _{x \times y \in \mathcal{S} \cap \Xi_k}\left|\hat{f}(x,y)-\hat{f}\left(\xi_{k, n}\right)\right| \leq \frac{C_2}{h^{d+2}}l_n = C_2\left(\frac{\ln(n)}{nh^{d+1}}\right)^{1/2} ;\\
    V3: &= \max _{1 \leq k \leq J_n} \sup _{x \times y \in \mathcal{S} \cap \Xi_k}\left|\mathbb{E}\left(\hat{f}\left(\xi_{k, n}\right)\right)-\mathbb{E}(\hat{f}(x,y))\right| \leq \frac{C_3}{h^{d+2}}l_n = C_3\left(\frac{\ln(n)}{nh^{d+1}}\right)^{1/2}.
\end{split}
\end{equation}
Without loss of generality, we assume that $C_4$ is a large value such that $C_4 > \max(C_2, C_3)$. Combining all pieces in results (\ref{Eq:V1V3final}) and (\ref{Eq:V2inequalityfinalleq}), we continue the inequality (\ref{Eq:Q1statrt}) to get a bound for $Q_1$ term, i.e.,
\begin{equation}\label{Eq:Q1fianl}
    \begin{split}
       Q_1 &:= \sup _{x \times y \in \mathcal{S}}|\hat{f}(x,y)-\mathbb{E}(\hat{f}(x,y))| \\
       & \leq  V_1+V_2+V_3\\
       & \leq 3C_4\left(\frac{\ln(n)}{nh^{d+1}}\right)^{1/2},~\text{with probability at least}~ 1 -  \frac{2J_n}{n^\alpha}.
    \end{split}
\end{equation}
More specifically, $J_n = O(\frac{1}{l_n^{d+1}}) = O\left( \left(\frac{n}{\ln(n)h^{d+3}}\right)^{\frac{d+1}{2}} \right)$. Thus, the probability $1 -  \frac{2J_n}{n^\alpha}$ converges to 1 and $3C_4\left(\frac{\ln(n)}{nh^{d+1}}\right)^{1/2}$ converges to 0 with an appropriately large constant $C_4$. 

Similarly, for $Q_3 := \sup_{x,y} \left|  \mathbb{E}(\hat{f}(x)f(y|x)) \right.$ $\left. - \hat{f}(x)f(y|x)  \right|$, we have
$$
\sup_{x,y} \left|  \mathbb{E}(\hat{f}(x)f(y|x)) - \hat{f}(x)f(y|x)  \right| \leq A_3 \sup_{x,y}| \hat{f}(x) -  \mathbb{E}(\hat{f}(x)) |;
$$
where $A_3 := \sup_{x,y}f(y|x)$ which is a constant under B2. Then, we can get 
\begin{equation}
    Q_3 := \sup_{x,y} \left|  \mathbb{E}(\hat{f}(x)f(y|x))  - \hat{f}(x)f(y|x)  \right| \leq 3A_3\cdot C^{\prime}_4 \left(\frac{\ln(n)}{nh^{d}}\right)^{1/2},~\text{with probability at least}~ 1 -  \frac{2J^{\prime}_n}{n^{\alpha^{\prime}}};
\end{equation}
where $C_4^{\prime}$ is another appropriately large constant; $J^{\prime}_n = O\left( \left(\frac{n}{\ln(n)h^{d+2}}\right)^{\frac{d}{2}} \right)$; $\alpha^{\prime} = C_4^{\prime} - A_2^{\prime}$; $A_2^{\prime}: = \sup_{x} K^2(x)$.

Lastly, we need to consider the order of $Q_2 := \sup_{x,y}\left| \mathbb{E}(\hat{f}(x,y)) - \mathbb{E}(\hat{f}(x)f(y|x))     \right|$ term. For $\mathbb{E}(\hat{f}(x,y))$, we have
\begin{equation}\label{Eq:Q11stterm}
\begin{split}
    \mathbb{E}(\hat{f}(x,y)) &= \frac{1}{h^{d+1}}\int \widetilde{K}( \Lambda^{-1}(u-z))f_Z(u)du \\
    & = \int \widetilde{K}(v)f_Z(\Lambda v + z)dv \\
    & = \int \widetilde{K}(v)f_Z(hv + z)dv \\
    & = \int \widetilde{K}(v)\left(f_Z(z) + hv^{T}\nabla f_Z(z)  + hv^{T} \nabla^2 f_Z(\bar{z})hv \right)dv \\
    & = f_Z(z) + h^2\int \widetilde{K}(v) v^{T} \nabla^2 f_Z(\bar{z})v dv.
\end{split}
\end{equation}
The second and third equality are due to the change of variable and the diagonal format of the matrix $\Lambda$; the fourth equality is due to the Taylor expansion under B3; the last equality is due to the property of the kernel function $\widetilde{K}(\cdot)$. Similarly, for $\mathbb{E}(\hat{f}(x)f(y|x))$, we have
\begin{equation}\label{Eq:Q12ndterm}
    \begin{split}
         \mathbb{E}(\hat{f}(x)f(y|x)) &= f(y|x))\frac{1}{h^d}\int K( \Lambda_x^{-1}(\mu - x))f(\mu)d\mu  \\
         & = f(y|x)\int K(\nu)f_X(\Lambda_x \nu + x)d\nu\\
         & = f(y|x)\int K(\nu)\left( f(x) + h\nu^T\nabla f(x) +  h\nu^T\nabla^2 f(\bar{x})h\nu \right)d\nu\\
         & = f_Z(z) + f(y|x)h^2\int K(\nu) \nu^T\nabla^2 f(\bar{x})\nu d\nu. 
    \end{split}
\end{equation}
Combining \cref{Eq:Q11stterm,Eq:Q12ndterm}, we can get
\begin{equation}\label{Eq:Q2final}
    \begin{split}
        Q_2 &= \sup_{x,y}\left| f_Z(z) + h^2\int \widetilde{K}(v) v^{T} \nabla^2 f_Z(\bar{z})v dv -  f_Z(z) - f(y|x)h^2\int K(\nu) \nu^T\nabla^2 f(\bar{x})\nu d\nu   \right| \\
        & = h^2\sup_{x,y} \left| \int \widetilde{K}(v) v^{T} \nabla^2 f_Z(\bar{z})v dv - f(y|x)h^2\int K(\nu) \nu^T\nabla^2 f(\bar{x})\nu d\nu \right|\\
        & \leq  h^2 \sup_{x,y} \left| \int \widetilde{K}(v) v^{T} \nabla^2 f_Z(\bar{z})v dv  \right|  +  \sup_{x,y} \left| f(y|x)h^2\int K(\nu) \nu^T\nabla^2 f(\bar{x})\nu d\nu \right| \\
        & = O(h^2).
    \end{split}
\end{equation}
The last equality is due to the assumption B2, B3 and D2. Then, we can define $Q_1 < 3C_4\left(\frac{\ln(n)}{nh^{d+1}}\right)^{1/2}$ and $Q_3 < 3A_3\cdot C^{\prime}_4 \left(\frac{\ln(n)}{nh^{d}}\right)^{1/2}$ as the event $E_1$ and $E_2$, respectively. We have $\mathbb{P}(E_1) \geq 1 -  \frac{2J_n}{n^\alpha}$ and $\mathbb{P}(E_2) \geq 1 -  \frac{2J^{\prime}_n}{n^{\alpha^{\prime}}}$. Therefore, with $\mathbb{P}(E_1 \bigcap E_2)$, we have
\begin{equation}
    \begin{split}
        \sup_{x,y}|\hat{f}(y|x) - f(y|x)| &\leq  C_1 \cdot ( Q_1 + Q_2 + Q_3)   \\
        & \leq C_1 \left( 3C_4\left(\frac{\ln(n)}{nh^{d+1}}\right)^{1/2} +  3A_3\cdot C^{\prime}_4 \left(\frac{\ln(n)}{nh^{d}}\right)^{1/2}\right) + O(h^2) \\
        & = O( C_4\left(\frac{\ln(n)}{nh^{d+1}}\right)^{1/2} + C^{\prime}_4 \left(\frac{\ln(n)}{nh^{d}}\right)^{1/2} + h^2).
    \end{split}
\end{equation}
For $\mathbb{P}(E_1 \bigcap E_2)$, by a simple equation that $\mathbb{P}(E_1\bigcup E_2) = \mathbb{P}(E_1) + \mathbb{P}(E_2) - \mathbb{P}(E_1 \bigcap E_2)$, we know that $\mathbb{P}(E_1 \bigcap E_2) \to 1$ in some rate since $\mathbb{P}(E_1)\to 1$, $\mathbb{P}(E_2) \to 1$, and $\mathbb{P}(E_1\bigcup E_2) \leq 1$. By defining $\mathbb{P}(E_1 \bigcap E_2)$ as $\kappa_n$, $\kappa_n$ will converge to 1 in some appropriate rate. Finally, we can verify our claim shown in the inequality (\ref{Eq:Kernelestimators}). 
\end{proof}

\begin{Theorem}[One variant of Theorem 3.1 in \cite{wu2024deep}]\label{Theorem:errorboundofC0f}
Under assumptions A1 to A3 and B1, let $\mathcal{F}_{\text{DNN}}$ be a class of standard fully connected feedforward DNN functions with width $W$ and depth $L$, respectively:
$$W := 3^{d+3} \max \left\{d\left\lfloor N_1^{1 / d}\right\rfloor, N_1+1\right\}~;~L: = 12 N_2+14+2d;$$ 
see \cref{Remakr:explanationsofWL} for the meaning of $W$ and $L$ for a DNN in this context. We take $N_1 = \ceil{ \frac{n^{\frac{d}{2(\tau+d)}}}{\log n}}$ and $N_2 = \ceil{ \log(n) } $; Then, for large enough $n$, i.e., $n> \max((2eM_1)^2, \text{Pdim}(\mathcal{F}_{\text{DNN}}))$, we have that $ \left\|\widehat{H} - H_0\right\|_{L^2(X)}^2\to 0$ with probability at least $1-\exp(-\gamma)$, as long as $\gamma = o(n)$. Under the further restriction of the modulus of continuity on $H_0$, we have that with probability at least $1-\exp(-n^{\frac{d+p-1}{\tau+d+p}})$:
\begin{equation}
   \left\|\widehat{H} - H_0\right\|_{L^2(X)}^2 \leq 
  C\cdot n^{-\frac{2}{\tau+d}} +  o(n^{-\frac{2}{\tau+d}});~ \text{for}~ d \geq 1; \tau>2; (d)\left\lfloor N_1^{1 / (d)}\right\rfloor \geq N_1+1;
\end{equation}
when $n$ is large enough, we have that with probability at least $1-\exp(-n^{\frac{d}{\tau+d}})$:
\begin{equation}
   \left\|\widehat{H} - H_0\right\|_{L^2(X)}^2 \leq C\cdot n^{-\frac{2}{\tau+d}} + o(n^{-\frac{2}{\tau+d}});~\text{for}~d\geq 1; \tau>2;(d)\left\lfloor N_1^{1 / (d)}\right\rfloor < N_1+1;
\end{equation}
where $C$ is a constant whose value may vary from context. In Theorem 3.1 of \cite{wu2024deep}, the dimension $d$ must be larger or equal to 2 since their estimation procedure is accompanied by a reference random variable. For the situation in which the input dimension is 1, we can prove a similar error bound with the requirement that $\tau$ is larger than 2. 
\end{Theorem}

\begin{Lemma}[Theorem 4.3 of \cite{shen2021deep}]\label{Lemma:DNNC0f}
Let $g$ be a uniformly continuous function defined on $E \subseteq[-M_2, M_2]^d$. For arbitrary $N_1 \in \mathbb{N}^{+}$ and $N_2 \in \mathbb{N}^{+}$, there exists a function $\phi$ implemented by a DNN with width $3^{d+3} \max \left\{d\left[N_1^{1 / d}\right\rfloor, N_1+1\right\}$ and depth $12 N_2+14+2 d$ such that
$$
\|g-\phi\|_{L^{\infty}(E)} \leq 19 \sqrt{d} \omega_g^E\left(2 M_2 N_1^{-2 / d} N_2^{-2 / d}\right);
$$
where $\omega_g^E(r)$ is the so-called modulus of continuity of $g$ on a subset $E$ belongs to the input space $\mathcal{X}$:
$$
\omega_g^E(r):=\sup \left\{\left|g\left(x_1\right)-g\left(x_2\right)\right|; d_\mathcal{X}\left(x_1, x_2\right) \leq r, x_1, x_2 \in E\right\},~\text{for any }~ r \geq 0.
$$
$d_\mathcal{X}(\cdot,\cdot)$ is a metric define in $\mathcal{X}$. 
\end{Lemma}

\section{\textsc{Complementary analysis }}\label{Appendix:B}

\subsection{Correction for monotonicity}\label{Subsec:CorrectionforMono}
An important prior step before performing cPI is correcting estimations $\{\widehat{H}_j(X_f)\}_{j=1}^g$ for monotonicity so that the subsequent procedure of building cPI works appropriately. 

In estimating the conditional distribution function, people usually take non-parametric local constant/linear techniques. Moreover, some adjustments have been proposed to fine-tune the kernel estimators, such as \cite{li2007nonparametric} applied a double-smoothing idea to make kernel estimators smooth w.r.t. the $Y$ variable and \cite{hansen2004nonparametric} took a straightforward way to monotonize the kernel estimation. Recently, \cite{das2020nonparametric} proposed several correction ideas for local linear methods and kept its advantage in estimating the CDF at boundary points. 

Motivated by the work of \cite{das2020nonparametric}, we propose three methods to correct (conditional or unconditional) CDF estimations at fixed grid points. Suppose we have CDF estimations $\{\widehat{F}_Z(z_j)\}_{j=1}^g$  for an arbitrary CDF function $F_Z$ evaluated at grid points $\{z_j\}_{j=1}^g$. We can correct values $\{\widehat{F}_Z(z_j)\}_{j=1}^g$ to guarantee the monotonicity in three simple ways as follows: 
\begin{itemize}
    \item[C1] Take $\widetilde{F}_{Z}(z_j) = \min(1, \max(0,\max_{1\leq k\leq j}\{\widehat{F}_Z(z_k)\}))$ for $j = 1, \ldots, g$. Then, we interpolate $\{\widetilde{F}_{Z}(z_j)\}_{j=1}^{g}$ by linear segments as the final estimation and denote it by $\widetilde{F}_{\text{LtoR}}$;

    \item[C2] Take $\widetilde{F}_{Z}(z_j) = \max(0, \min(1,\min_{j\leq k\leq g}\{\widehat{F}_Z(Z_k)\}))$ for $j = g, \ldots, 1$. Then, we interpolate $\{\widetilde{F}_{Z}(z_j)\}_{j=1}^{g}$  by linear segments as the final estimation and denote it by $\widetilde{F}_{\text{RtoL}}$;

    \item[C3] Take the average of the above two correction results, i.e., $\widetilde{F}_{\text{Avg}}(z_j) := (\widetilde{F}_{\text{RtoL}}(z_j) + \widetilde{F}_{\text{LtoR}}(z_j))/2$, for $j = 1,\dots, g$. Then, we make the interpolation to get the final estimation, namely $\widetilde{F}_{\text{Avg}}$.
\end{itemize}



Taking a closer look at these three correction methods, we claim that $\widetilde{F}_{\text{LtoR}}$ will always be larger than or equal to $\widetilde{F}_{\text{RtoL}}$ for all grid points $\{z_j\}_{j=1}^g$. Thus, $\widetilde{F}_{\text{Avg}}$ lies between $\widetilde{F}_{\text{RtoL}}$ and $\widetilde{F}_{\text{LtoR}}$. We summarize this observation in the proposition below.
\begin{Proposition}\label{Prop:Correcmethods}
    $\widetilde{F}_{\text{Avg}}(z_k)$ always lies between $\widetilde{F}_{\text{LtoR}}(z_k)$ and $\widetilde{F}_{\text{RtoL}}(z_k)$, i.e.,   $\widetilde{F}_{\text{RtoL}}(z_k)\leq \widetilde{F}_{\text{Avg}}(z_k) \leq \widetilde{F}_{\text{LtoR}}(z_k)$ for $k = 1,\ldots, g$. 
\end{Proposition}

\begin{proof}[\textsc{\textbf{Proof of \cref{Prop:Correcmethods}}}]
We consider two situations. If $\{\widehat{F}_Z(z_j)\}_{j=1}^g$ are monotonic w.r.t. $j$, the correction for monotonicity does not play a role. Thus, $\widetilde{F}_{\text{Avg}}(z_k) \equiv \widetilde{F}_{\text{LtoR}}(z_k) \equiv  \widetilde{F}_{\text{RtoL}}(z_k)$ for $k = 1,\ldots, g$. 
    
Otherwise, $\widetilde{F}_{\text{LtoR}} \geq \widetilde{F}_{\text{RtoL}}$ at all points $z_1,\ldots,z_g$. Take the values of $\widetilde{F}_{\text{LtoR}}$ and $\widetilde{F}_{\text{RtoL}}$ at the point $z_1$ as an example, by their definition, $\widetilde{F}_{\text{LtoR}}(z_1) = \min(1, \max(0,\max_{1\leq k\leq 1}\{\widehat{F}_Z(z_k)\}))$ and $\widetilde{F}_{\text{RtoL}}(z_1) = \max(0, \min(1,\min_{1\leq k\leq g}\{\widehat{F}_Z(z_k)\}))$. Obviously, $\min_{1\leq k\leq g}\{\widehat{F}_Z(z_k)\} \leq \max_{1\leq k\leq 1}$ $\{\widehat{F}_Z(z_k)\}$. Then $\widetilde{F}_{\text{RtoL}}(z_1)\leq \widetilde{F}_{\text{LtoR}}(z_1)$. Similarly, we can show such a relationship between $\widetilde{F}_{\text{RtoL}}$ and $\widetilde{F}_{\text{RtoL}}$ for other grid points. 
\end{proof}

In practice, we will routinely apply these three correction methods to create the monotonicity for our conditional CDF estimations. More specifically, according to the C1 approach, we can take $\widetilde{F}_{Y|X_f}(q_j) = \min(1, \max(0,\max_{1\leq k\leq j}$ $\{\widehat{H}_k(X_f)\}))$ for $j = 1, \ldots, g$. Then, we interpolate $\{\widetilde{F}_{Y|X_f}(q_j)\}_{j=1}^{g}$ by linear segments as the final estimation.

\begin{Remark}[Comparison of our correction methods and existing ones]\label{Remark:Comparion}
C1 and C2 approaches are based on the idea proposed by \cite{das2020nonparametric}. In their work, they also advocated a method to correct for monotonicity of CDF estimation by positivizing the density function derived from a double-smoothing local linear estimator and taking the integral and normalizing to recover the estimation of CDF. In our case, due to the discreteness of our estimations, we do not need to repeat such an approach upon taking a derivative and then solving an integral. It turns out that these three correction methods work well in simulation and empirical studies. Moreover, we provide a toy example to illustrate the effects of these three methods on the monotonicity correction in \cref{B1:correctioneffects}. 
\end{Remark}

\subsection{The effects of three correction methods in \cref{Subsec:CorrectionforMono}}\label{B1:correctioneffects}
Denote the simple interpolation of $\{\widehat{H}_j(X_f)\}_{j=1}^g$ on one specific simulated dataset based on Model-1 presented in \cref{Sec:Simulation} by $\widehat{F}$. We plot the oracle CDF function $F$ obtained by Monte Carlo simulation, three corrected CDF estimations ($\widetilde{F}_{\text{LtoR}}$, $\widetilde{F}_{\text{RtoL}}$ and $\widetilde{F}_{\text{Avg}}$) and the naive one $\widehat{F}$ in \cref{fig:effectsofcorrection}. As we can find, the naive CDF estimation $\widehat{F}$ fluctuates severely and the cross-quantile phenomenon appears frequently, i.e., $\widehat{F}(q_j) < \widehat{F}(q_i)$ for $q_j > q_i$. On the other hand, three corrected estimations work well. Moreover, $\widetilde{F}_{\text{Avg}}$ is always located between the curve of $\widetilde{F}_{\text{LtoR}}$ and $\widetilde{F}_{\text{RtoL}}$ and it mimics the oracle CDF curve well, especially for two tails parts. 

\begin{figure}[htbp]
    \centering
    \includegraphics[width=1\linewidth]{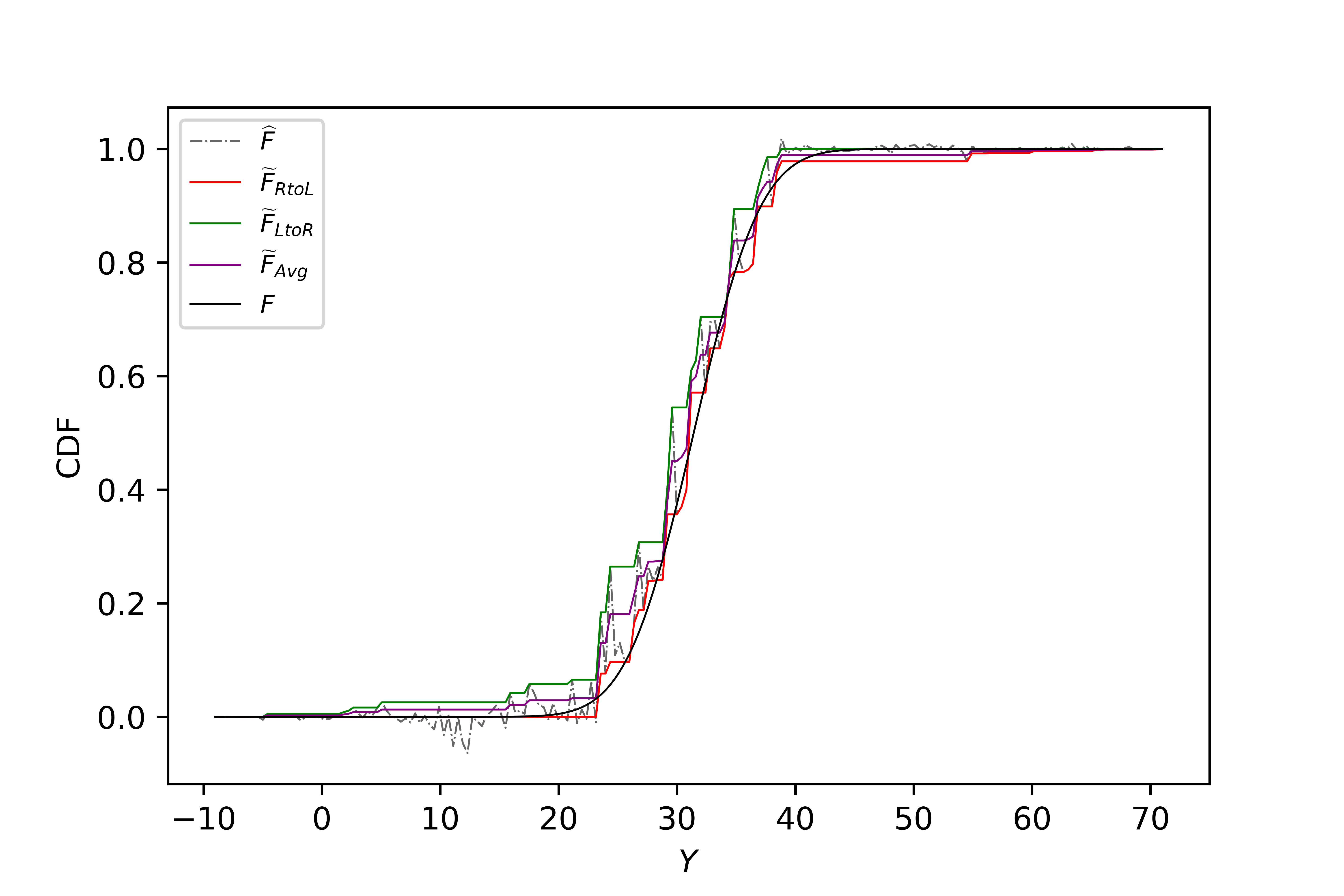}
    \caption{The effects of three monotonicity corrections methods.}
    \label{fig:effectsofcorrection}
\end{figure}
\FloatBarrier
\subsection{The density estimation with the maximum likelihood and least square cross-validation methods}\label{B2:cross-validationmethod}
 The coding of the conditional density estimation relies on two Python packages, \textit{Conditional Density Estimation} and \textit{statsmodels}. The corresponding methodology comes from Section 5.2 of \cite{li2007nonparametric}. In detail, the least square cross-validation method attempts to minimize the integrated square error (ISE):
 $$
\begin{aligned}
I S E & =\int\{\hat{f}(y | x)-f(y | x)\}^2 f(x)  d x d y \\
& =I_{1 n}-2 I_{2 n}+I_{3 n};
\end{aligned}
$$

where,

$$
I_{1 n}=\int \hat{f}(y | x)^2 f(x)  d x d y; I_{2 n}=\int \hat{f}(y | x) f(x, y)  d x d y; I_{3 n}=\int f^2(y | x) f(x) d x d y.
$$
Actually, we just need to minimize $I_{1n} - 2I_{2n}$ since $I_{3n}$ dose not dependent the smoothing bandwidth. Assuming the dependent variable is also continuous and applying the product kernel and univariate Gaussian kernel functions $k(\cdot)$ and $k_0(\cdot)$ as we did in the main text, we can rewrite the $I_{1 n}$ as:
\begin{equation}\label{Eq:cvlsI1n}
    \begin{split}
        I_{1 n} &= \int \hat{f}(y | x)^2 f(x)  d x d y \\
        & = \int \frac{\frac{1}{n^2|\Lambda_x|^2 h^2} \sum_{i_1=1}^n \sum_{i_2=1}^n K\left(\Lambda_x^{-1}(x-X_{i_1})\right) K\left(\Lambda_x^{-1}(x-X_{i_2})\right)  k_0\left(\frac{y-Y_{i_1}}{h}\right)  k_0\left(\frac{y-Y_{i_2}}{h}\right)}{\hat{f}(x)^2} f(x)  d x d y \\
        & = \int \hat{G}(x) \frac{f(x)}{\hat{f}(x)^2}  d x\\
        & =\mathbb{E}_X\left[\frac{\hat{G}(X)}{\hat{f}(x)^2} \right];\\
    \end{split}
\end{equation}
where $\hat{f}(x)$ is the marginal density estimator; $\hat{G}(x)$ is defined as below:
$$\hat{G}(x) : = \frac{1}{n^2|\Lambda_x|^2 h^2} \sum_{i_1=1}^n \sum_{i_2=1}^n K\left(\Lambda_x^{-1}(x-X_{i_1})\right) K\left(\Lambda_x^{-1}(x-X_{i_2})\right) \int k_0\left(\frac{y-Y_{i_1}}{h}\right)  k_0\left(\frac{y-Y_{i_2}}{h}\right) dy.
$$
For the last expectation term w.r.t. $X$ in \cref{Eq:cvlsI1n}, we can estimate by the leave-one-out estimators, i.e., approximating $I_{1n}$ by $\hat{I}_{1 n}$ which is defined below:
$$\hat{I}_{1 n}=\frac{1}{n} \sum_{i=1}^n \frac{\hat{G}_{-i}\left(X_i\right)}{\hat{f}_{-i}\left(X_i\right)^2}.$$
For example, 
$$
\hat{f}_{-i}\left(X_i\right) = \frac{1}{(n-1) |\Lambda_x|} \sum_{j=1, j \neq i}^n K\left(\Lambda_x^{-1}(X_i-X_j)\right).
$$
In practice, the denominator of $\hat{I}_{1 n}$ is a squared term of $\hat{f}_{-i}\left(X_i\right)$ that could be 0 numerically when the sample size is small. Then, the optimization procedure to find the optimal bandwidth with \textit{optimize.fmin} function in Python is unstable. On the other hand, for the maximum likelihood cross-validation, we do not have such a squared term in the denominator, so it is more numerically stable than the least square cross-validation method. For a complete comparison, the objective function of maximum likelihood cross-validation is presented below:
$$
\mathcal{L}=\sum_{i=1}^n \ln \hat{f}_{-i}\left(Y_i \mid X_i\right) = 
 \sum_{i=1}^n \ln  \left(\hat{f}_{-i}\left(X_i, Y_i\right) / \hat{f}_{-i}\left(X_i\right) \right);
$$
where $\hat{f}_{-i}\left(X_i, Y_i\right)$ and $\hat{f}_{-i}\left(X_i\right)$ are leave-one-out kernel estimators of $f\left(X_i, Y_i\right)$ and $f\left(X_i\right)$, respectively.

\vskip 0.2in
\bibliography{sample}

\end{document}